\documentclass[journal,twoside,web]{ieeecolor}
\usepackage{generic}
\usepackage{cite}
\usepackage{amssymb,amsmath,amsopn,amsfonts,colordvi}
\usepackage{mathtools}
\usepackage{comment}
\usepackage{mathrsfs}
\usepackage{algorithmic}
\usepackage{graphicx}
\usepackage{textcomp}
\usepackage{subfig}
\usepackage{algorithm}
\newtheorem{theorem}{Theorem}
\newtheorem{problem}{Problem}
\newtheorem{proposition}{Proposition}

\newtheorem{lemma}{Lemma}
\newtheorem{remark}{Remark}
\newtheorem{definition}{Definition} 
\newtheorem{assumption}{Assumption}

\def\BibTeX{{\rm B\kern-.05em{\sc i\kern-.025em b}\kern-.08em
		T\kern-.1667em\lower.7ex\hbox{E}\kern-.125emX}}
\begin{document}
	\title{Adaptive Incentive Design for Linear-Quadratic Dynamic Games:  
    No-Regret Social Coordination}
	\author{Lantian Zhang, Yibei Li, Georgios Vasileiou, Lihua Xie, \IEEEmembership{Fellow, IEEE}, and Silun Zhang
		\thanks{This work is supported by The Wallenberg AI, Autonomous Systems and Software Program (WASP) which is
funded by the Knut and Alice Wallenberg Foundation. (Corresponding author: Yibei Li).}
		\thanks{Lantian Zhang, Georgios Vasileiou, and Silun Zhang are with the Department of Mathematics, KTH Royal Institute of Technology, Stockholm 10044, Sweden (e-mails: lantian@kth.se, geovas@kth.se, silunz@kth.se). 
        Yibei Li is with the State Key Laboratory of Mathematical Sciences,
Academy of Mathematics and Systems Science, Chinese Academy of
Sciences, Beijing 100190, China (e-mail: yibeili@amss.ac.cn). 
Lihua Xie is with the School of Electrical and Electronic Engineering,
Nanyang Technological University, Singapore 639798 (e-mail: elhxie@ntu.edu.sg).
        }}
	
	\maketitle
	
	\begin{abstract}
Incentives can align private interests with collective goals, but their design becomes intertwined with learning when agents’ preferences are unknown. This paper develops an adaptive incentive-design framework for multi-agent linear–quadratic dynamic games, where a system planner can simultaneously learn agents’ private preferences and steer their strategic behavior towards the social optimum.
To this end, we first characterize the complete family of affine incentives that implement the social optimum as the unique open-loop Nash equilibrium while ensuring exact budget balance.
When agents’ preferences are unknown, we develop an online algorithm to estimate their preference parameters and establish strong consistency of the estimates under a weak trajectory-dependent excitation condition, without requiring persistent excitation or i.i.d. trajectories.
We then integrate this estimator into an active-learning incentive mechanism that combines certainty-equivalent coordination with decaying random probing. We prove that the probing generates sufficient information through the agents’ equilibrium responses and the system dynamics, ensuring consistent preference learning from the resulting closed-loop trajectories. We show that the induced trajectories converge almost surely to their socially optimal counterparts, while cumulative social-cost regret relative to the full-information oracle and cumulative net-budget imbalance grow sublinearly.  
Simulations of a transboundary pollution game demonstrate the effectiveness of the proposed framework.	
	\end{abstract}
	

\section{Introduction}
Cyber-physical systems increasingly comprise strategic decision makers, including humans and AI agents, whose actions are shaped by individual objectives and strategic responses to one another and their environment. 
Individually rational decisions may therefore conflict with collective goals, which makes system-level coordination a fundamental challenge.
\textit{Incentive design} addresses this challenge by allowing a system planner to design and commit to a payment mechanism that aligns agents' private incentives with desired system-level objectives (see \cite{ratliff_review} and references therein). This paradigm has shown its effectiveness in various domains, including smart grids, where pricing and rewards promote efficient energy consumption \cite{pricing_energy, rahi}; transportation systems, where congestion tolls influence travelers' route choices and improve traffic efficiency \cite{trans}; and environmental governance, where carbon taxes and emission pricing encourage firms to reduce emissions and promote sustainable development~\cite{pollution_li, pollution_franciso}. 

A core obstacle in incentive design is the \textit{information asymmetry} between the system planner and strategic agents, where agent preferences are rarely known in advance~\cite{contract, maskin1984monopoly}. This is because agents pursuing private interests are often unwilling to disclose their preferences and, in some settings, may themselves lack complete knowledge of their own objective functions. Adaptive incentive-design methods address such asymmetry by learning agents’ preferences from their responses while simultaneously adjusting the incentive policy~\cite{ratliff}. Under the assumption that agents follow asymptotically convergent equilibrium-seeking dynamics, \cite{learning_agent} develops a two-timescale adaptive algorithm that couples preference learning with incentive adaptation. By contrast, no-regret adaptive incentive design methods~\cite{vasileiou2026adaptiveincentivedesignregret} integrate learning and incentive selection within a single-timescale loop. They achieve vanishing average regret to a foresighted oracle via an explicit exploration-exploitation trade-off. Other approaches address incomplete preference knowledge in incentive design through social-gradient methods~\cite{vasileiou2026incentive_cdc}, without explicitly estimating private costs or equilibrium sensitivities, and also under boundedly rational agent responses~\cite{active-inverse}.

Adaptive incentive design becomes even more relevant and challenging in dynamic environments, where current actions alter a shared environment state and thereby shape agents' future costs and strategic responses. The planner must therefore coordinate entire state–action trajectories, rather than isolated one-shot decisions, while balancing immediate behavioral effects against their long-term consequences. Incentive problems in such games with dynamic players and environments occur in many applications, including, e.g., transboundary pollution games~\cite{pollution_li}, building energy pricing~\cite{ratliff_energy}, and transportation games~\cite{trans,manxi_wu}. Linear–quadratic (LQ) dynamic games provide a canonical setting for this problem, where linear dynamics capture action propagation through time, and quadratic objectives encode heterogeneous valuations of system states and control effort while retaining a tractable Nash structure. 

This paper develops an adaptive incentive design framework for multi-agent LQ dynamic games that actively learns agents' unknown preferences from noisy responses and simultaneously steers the induced Nash equilibrium toward the social optimum. The main contributions of the paper are threefold:
\begin{enumerate}
    \item[1)] First, we formulate stochastic adaptive incentive design for multi-agent LQ dynamic games and, under full information, characterize the complete family of affine incentives that implement the social optimum and promote budget balance (Prop.~\ref{pro1}).
    \item[2)] Second, we propose an online learning algorithm 
    to estimate agents’ unknown preference parameters and, by martingale convergence theory, prove its strong consistency without requiring persistent excitation (PE) or i.i.d. trajectory data (Thm. 1). This relaxes standard identifiability conditions in cost learning
    (see, e.g., \cite{li2024inverse,inverse_control, QU20232134}), and enables preference learning and incentive design in a closed loop.
    \item[3)] Third, we develop an adaptive incentive scheme based on certainty equivalence and decaying random probing. We establish that the probing propagates through agents’ equilibrium responses and the system dynamics to generate sufficient excitation (Lem.~\ref{lem:adaptive_excitation}, \& Thm.~\ref{thm2}), while the resulting open-loop equilibrium trajectories converge to the social optimum. The proposed scheme achieves vanishing average regret relative to the full-information oracle and vanishing average net budget (Thm.~\ref{thm3}).
\end{enumerate}

\subsection{Related Work}

{\noindent\bf Incentive Design and Stackelberg Games.} 
Classical control-theoretic studies for incentive design established incentive controllability~\cite{HoLuhMuralidharan1981,Ho1982}, and affine incentive construction under different information structures~\cite{ZhengBasar1982}. Related work developed dynamic and LQ leader–follower formulations~\cite{ChangHo1981,abou,freiling}, while functional incentive design in LQ systems has also been studied using bilinear incentives~\cite{MattBolognaniDorfler2026}. An interesting recent work \cite{leon2025onlinelearningdynamicvickreyclarkegroves} studies dynamic VCG mechanisms in unknown Markov decision processes. With known agents' preferences, these results provide the structural foundation for incentive problems. For unknown follower utilities, Lauffer et al.~\cite{no-regret_stackelberg} developed a no-regret learning algorithm for a finite-state Markov model with a single follower. Our focus is instead on continuous-state, multi-agent LQ dynamic games, where the planner adapts affine payments to asymptotically implement socially optimal state–action trajectories. 

{\noindent\bf Inverse Optimal Control and Inverse Games.} 
Preference learning in dynamic decision-making process is closely related to inverse optimal control, initiated by Kalman~\cite{inverse_kalman}, which infers cost parameters from observed optimal behavior. In LQ settings, existing work investigates cost reconstruction and identifiability~\cite{inverse_control,li2020continuous,Inverse_qu}, including partial and noisy observations~\cite{li2024inverse,zhang2021discrete,inverse_yu}. Inverse-game methods extend this perspective to cost identification through observed Nash equilibrium behavior~\cite{inverse_lq_game,inverse_game_Observations,zhang2019non,Inverse_Reinforcement}. Beyond cost identification from given trajectories, this paper studies a learning-and-decision closed-loop setting in which each incentive depends on past estimates and shapes the trajectories used for subsequent learning. We account for both the history dependence and regression bias induced by observation noise and establish strong consistency under a trajectory-dependent excitation condition.
The proposed adaptive incentive scheme then generates sufficient excitation to satisfy this condition.

{\noindent\bf Game-Based Control Systems.} 
Game-based control systems describe a regulator influencing a lower-level noncooperative dynamic game, with controllability and stabilizability studied in~\cite{renren_zhang19,renren_zhang21}. More closely related, Liu and Guo~\cite{liu_siam} considered adaptive stabilization under unknown system dynamics, with the regulator and agents learning the dynamics and adjusting their strategies online. Our setting differs in both the source of uncertainty and the regulation objective: the physical dynamics are known, whereas the agents’ cost parameters are private, and incentives are designed to align the induced Nash equilibrium with a prescribed social optimum.

{\noindent\bf Adaptive Control.} Our framework is structurally related to adaptive control, where recursive identification is coupled with online adaptation of the decision. Classical self-tuning regulation combines recursive identification with certainty-equivalent or diminishing-excitation control \cite{Astrom1973,Goodwin1981,Kumar1990,Guo1995}, with extensions to adaptive LQG \cite{Duncan1999}, data-driven LQ design \cite{DorflerTesiDePersis2023}, and nonlinear systems~\cite{ZhangGuo2026Tracking, Zhang2026}.
Unlike direct control of a system with unknown dynamics, however, the planner here chooses payments rather than the agents’ control actions. Probing must therefore propagate through the equilibrium response before it excites the cost regression. Our analysis quantifies this indirect excitation and its consequences for preference-learning accuracy and social-cost regret.

\subsection{Notations}
 For a positive integer $L$, we write $[L]=\{1,2,\ldots,L\}$.
For a matrix $X\in\mathbb{R}^{p\times q}$, $\operatorname{vec}(X)\in\mathbb{R}^{pq}$ denotes the vectorization of $X$ formed by stacking its columns. For a symmetric matrix $X\in\mathbb{R}^{p\times p}$, $\operatorname{vech}(X)\in\mathbb{R}^{p(p+1)/2}$ denotes the half-vectorization of $X$, obtained by stacking only the entries on and below the diagonal; the associated duplication matrix $\mathscr{D}_p\in\mathbb{R}^{p^2\times p(p+1)/2}$ is the unique matrix satisfying
$
\operatorname{vec}(X)=\mathscr{D}_p\operatorname{vech}(X)
$
for every symmetric $X\in\mathbb{R}^{p\times p}$. For a positive integer $n$, let $\mathbf{1}_n$ denote the
$n$-dimensional column vector of ones and $I_n$ the
$n\times n$ identity matrix.
Given a positive integer $s$, 
$X\in\mathbb{R}^{p\times p}$ and $Y\in\mathbb{R}^{p\times q}$, the lower-triangular block Toeplitz matrix
$\mathcal{L}s(X,Y)$ is defined by 
\begin{equation}\label{eq: operator_L}
\mathcal{L}_s(X, Y)
=
\begin{bmatrix}
Y &        &        &  \\
\vdots & \ddots &        &  \\
X^{s-1}Y & \cdots & Y
\end{bmatrix}\in\mathbb{R}^{sp\times sq}.    
\end{equation}
For any positive integers $r$ and $s$, and any
$z=[z_1^{\top},\ldots,z_r^{\top}]^{\top}\in\mathbb{R}^{rs}$ with
$z_i\in\mathbb{R}^s$, define the operator
$\mathcal{K}_{s,r}: \mathbb{R}^{rs}\to \mathbb{R}^{rs\times s^2}$ by
\begin{equation}\label{eq:operator}
    \mathcal{K}_{s,r}(z)
    =
    \begin{bmatrix}
        z_1^{\top}\otimes I_s\\
        \vdots\\
        z_r^{\top}\otimes I_s
    \end{bmatrix}.
\end{equation}

The remainder of the paper is organized as follows. Section~\ref{sec2} introduces the incentivized dynamic game model. Section~\ref{sec3} formulates the adaptive incentive design problem for LQ games.
Section~\ref{sec4} presents the main results on preference
estimation, including an online algorithm for estimating the agents' unknown preference parameters and an analysis of its convergence. Section~\ref{sec5} develops the adaptive incentive design algorithm and establishes vanishing average regret. Section~\ref{sec7} presents the proofs of the main results. Section~\ref{sec:sim} illustrates an application scenario through numerical examples. Finally, Section~\ref{se6} concludes the paper.

  \section{Preliminaries: Incentivized LQ Dynamic Games}\label{sec2} 

Consider $L$ strategic \textit{agents} coupled through a common linear dynamical system
\begin{equation}\label{eq2}
\begin{cases}
&
x(t+1)=Ax(t)+\sum_{i=1}^{L}B_{i}u_{i}(t), \;\; t=0, \dots, T-1,
\\
&x(0)=\xi,
\end{cases}
\end{equation}
where $x(t)\in \mathbb{R}^{n}$ is the system state at time $t$, $\xi\in \mathbb{R}^{n}$ is the initial state, and $u_{i}(t)\in \mathbb{R}^{m}$ is the strategy of agent $i$. $A \in \mathbb{R}^{n\times n}$ and $B_i \in \mathbb{R}^{n\times m}$ are deterministic system matrices. 
 Stacking agent strategy $\boldsymbol{u}_i=[u_i(0)^\top,\ldots,u_i(T-1)^\top]^\top \in\mathbb{R}^{mT}$, each agent $i\in[L]$ is associated with a private quadratic cost function
\begin{equation}\label{eq:intrinsic_cost}
\begin{aligned}
    c_i(\xi, \boldsymbol{u}_i, \boldsymbol{u}_{-i})
    =\sum_{t=0}^{T-1}\big[(x(t)-x_i^{d})^{\top}Q_i(x(t)-x_i^{d})&\\
    +u_{i}(t)^{\top}R_iu_{i}(t)&\big], 
    \end{aligned}
\end{equation}
where $\boldsymbol{u}_{-i}=(\boldsymbol{u}_{1}, \dots, \boldsymbol{u}_{i-1}, \boldsymbol{u}_{i+1},\dots, \boldsymbol{u}_{L})$ denotes the strategy profile of all agents except $i$. The matrices $Q_i\succeq 0$ and $R_i\succ 0$, together with the target state $x_i^{d}\in\mathbb{R}^n$, are the private preference parameters of agent $i$.

A \emph{system planner} (principal) seeks to guide the agents' strategic behavior but cannot directly control their strategies. Instead, the planner announces an incentive law 
$\gamma_i\!:\!\mathbb{R}^{n}\!\times\!\mathbb{R}^{m}\!\times \mathbb{R}
\!\to\mathbb{R}$ for each agent $i$. The incentive
$\gamma_i(x(t), u_i(t),t)$ specifies a monetary payment to agent $i$ based on the shared state $x(t)$ and the agent's action $u_i(t)$, with positive values representing taxes (penalties) and negative values representing subsidies (rewards). The incentive law $\gamma_i$ is parameterized as:
\begin{equation}\label{incentive}
\gamma_i(x(t),u_{i}(t), t)
=
p_i(t)^{\top}
\begin{bmatrix}
x(t)\\
u_{i}(t)\\
1
\end{bmatrix},    
\end{equation}
where $p_i(t)\in\mathbb{R}^{n+m+1}$
is the incentive parameter vector, consisting of the state and action payment rates and a constant offset.
Similar affine incentive laws are also used in \cite{learning_agent, pollution_li}.
Stacking incentive parameters $\boldsymbol{p}_i=[p_i(0)^\top,\ldots,p_i(T-1)^\top]^\top \in\mathbb{R}^{(n+m+1)T}$, the effective (total) cost of agent $i$ under the incentive $\boldsymbol{p}_i$ is given by
\begin{equation}\label{eq:total_cost}
\begin{aligned}
J_i(\xi,\boldsymbol{u}_i, \boldsymbol{u}_{-i},\boldsymbol{p}_i)={}\sum_{t=0}^{T-1} \Big[
    (x(t)-x_i^d)^\top Q_i (x(t)-x_i^d)& \\
    + u_i(t)^\top R_i u_i(t)
    + \gamma_i\bigl(x(t),u_i(t),t\bigr)&
    \Big].
\end{aligned}
\end{equation}

Define the stacked strategy and incentive profiles of all agents, respectively, as $\boldsymbol{U}=[\boldsymbol{u}_1^{\top},\cdots,\boldsymbol{u}_L^{\top}]^{\top}\in \mathbb{R}^{mLT}$ and $\boldsymbol{P}=[\boldsymbol{p}_1^{\top},\ldots,\boldsymbol{p}_L^{\top}]^{\top}\in\mathbb{R}^{(n+m+1)LT}$.
Given $\xi$ and $\boldsymbol{P}$, the multi-agent interaction with dynamics~\eqref{eq2} and cost~\eqref{eq:total_cost} defines an \emph{incentivized dynamic game}, whose special case $\boldsymbol{P}=0$ recovers the classical LQ dynamic game. 
The Nash equilibrium of the incentivized dynamic game is characterized
as follows. 
\begin{definition}[Incentivized Open-Loop Nash Equilibrium]\label{de1}
Given $\xi\in \mathbb R^{n}$ and $\boldsymbol{P}\in\mathbb{R}^{(n+m+1)LT}$, in the incentivized dynamic game with agent cost \eqref{eq:total_cost} and dynamics \eqref{eq2},
a joint strategy profile $\boldsymbol{U}^{\diamond}=[(\boldsymbol{u}_{1}^{\diamond})^{\top},\cdots, (\boldsymbol{u}_{L}^{\diamond})^{\top}]^{\top}\in \mathbb{R}^{mLT}$ is an incentivized open-loop Nash equilibrium (IONE) strategy induced by $(\xi, \boldsymbol{P})$ if, for every agent $i\in [L]$,
\begin{equation}\nonumber
J_{i}(\xi, \boldsymbol{u}_{i}^{\diamond}, \boldsymbol{u}_{-i}^{\diamond}, \boldsymbol{p}_i)\leq J_{i}(\xi, \boldsymbol{u}_{i},\boldsymbol{u}_{-i}^{\diamond}, \boldsymbol{p}_i),\quad \forall \boldsymbol{u}_{i}\in\mathbb{R}^{mT}.
\end{equation}
\end{definition}
\begin{assumption}\label{assum_ne}
When the incentive is $0$, i.e., \(\boldsymbol{P}=0\),
for every initial state $\xi\in\mathbb{R}^{n}$, the game defined by \eqref{eq2} and \eqref{eq:total_cost} admits a unique open-loop NE.
\end{assumption}
\begin{remark}
By the classical results in LQ dynamic games (see Theorem~6.2, p.~243, in~\cite{Basar}), Assumption~\ref{assum_ne} holds if the matrix
$\Delta(t)$ is nonsingular for every $t=0,\ldots,T-1$, where
$\Delta(t)$ is recursively defined backward in
time by
\vspace{-4pt}
\[
\begin{gathered}
\Delta(t)
=
I_n+\sum_{i=1}^{L}
B_iR_i^{-1}B_i^\top V_i(t+1),
\quad t=T-1,\ldots,0,
\\[-3pt]
V_i(t)
=
Q_i+A^\top V_i(t+1)\Delta(t)^{-1}A,
\; V_i(T)=0,
\; i=1,\ldots,L.
\end{gathered}
\]
\end{remark}
\vspace{4pt}

Allowing the initial state $\xi$ and the incentive parameter $\boldsymbol{P}$ to vary over $\mathbb{R}^n\times\mathbb{R}^{(n+m+1)LT}$ yields a unique \emph{IONE response map} that associates each pair \((\xi,\boldsymbol{P})\) with its induced Nash equilibrium, as formalized below. 
\begin{lemma}\label{lem_ne}
Under Assumption~\ref{assum_ne}, there exists a unique map $\mu^\diamond:\mathbb{R}^n\times\mathbb{R}^{(m+n+1)LT}\to\mathbb{R}^{mLT}$ such that $\mu^{\diamond}(\xi,\boldsymbol{P})$ is the unique IONE strategy  induced by $(\xi,\boldsymbol{P})$. Specifically,
\begin{equation}\label{first-order}
 \mu^\diamond(\xi,\boldsymbol{P})
=
-M^{-1}\bigl(\mathcal H\boldsymbol{P}+\eta\bigr),   
\end{equation}
where $M$, $\mathcal H$, and $\eta$ are explicitly defined in \eqref{uniquness}--\eqref{eta} of Appendix \ref{sec:proofs}.
\end{lemma}

The proof of Lemma \ref{lem_ne} is given in Appendix \ref{sec:proofs}.
\begin{remark}
We remark that the variables $\xi$ and $\boldsymbol{P}$ may be random in the adaptive setting considered later. In this case, $\mu^{\diamond}(\xi, \boldsymbol{P})$
is defined samplewise, with Definition~\ref{de1} applied to each realized pair of
$(\xi, \boldsymbol{P})$.  
\end{remark}

Given the initial state $\xi$ and the agents' strategy profile $\boldsymbol{U}$, the system planner evaluates the overall
performance of the system through a \emph{social cost}
\begin{equation}\label{social_cost}
\begin{aligned}
\Phi(\xi, \boldsymbol{U})
=\sum_{t=0}^{T-1}\bigg[
(x(t)-x^{d*})^\top Q^{*}(x(t)-x^{d*})&\\[-5pt]
+
\sum_{i=1}^{L}u_{i}(t)^\top R_i^{*} u_{i}(t)&
\bigg],       
\end{aligned}
\end{equation}
where $x^{d*}\in\mathbb{R}^n$ is the prescribed social target,
$Q^*\succeq0$, and $R_i^*\succ0$ denote the weighting matrices. The first term penalizes deviations of
the system state from the social target, while the second accounts for the
cost of the agents' control actions.

Incentive design allows the planner to steer the agents' behavior toward minimizing the social cost by choosing a suitable incentive $\boldsymbol{P}$. We formulate this in the next section.
\section{PROBLEM FORMULATION: ADAPTIVE INCENTIVE
DESIGN}\label{sec3}
In this section, we introduce the adaptive incentive design (AID) problem for LQ dynamic games.

The $L$ agents and the planner engage in a sequential multi-round interaction process, with rounds indexed by $h=1,2,\ldots$. Each round is referred to as a \emph{season}. Let $(\Omega , \mathcal{F}, \mathbb{P})$ be a complete probability
space, equipped with a filtration
$\{\mathcal{F}_h\}_{h\geq0}$, i.e., $\mathcal{F}_0\subset \mathcal{F}_1 \subset \cdots$, where $\mathcal{F}_h$ represents the information
 available to the planner at the end of season $h$,
and \(\mathcal{F}_0\) denotes the initial information before season \(1\).

In each season $h\geq 1$, the interaction between the planner and agents is outlined as follows:

(i) Incentive selection. At the beginning of season $h$, the planner observes an $\mathcal{F}_{h-1}$-measurable initial state $\xi_h$ and, using the agents' trajectory observations from previous seasons,
selects an $\mathcal{F}_{h-1}$-measurable incentive sequence 
$\boldsymbol{p}_{i, h}=[p_{i,h}(0)^{\top},\ldots,p_{i,h}(T-1)^{\top}]^{\top}\in\mathbb{R}^{(m+n+1)T}$ for each agent $i$, where \(p_{i,h}(t)\in\mathbb R^{m+n+1}\) denotes the incentive parameter at time \(t\). The incentives of all agents are stacked as
$\boldsymbol{P}_{h}=[\boldsymbol{p}_{1, h}^{\top},\cdots, \boldsymbol{p}_{L, h}^{\top}]^{\top}\in\mathbb{R}^{(m+n+1)LT}$. 

(ii) Agents' response. After observing $(\xi_{h}, \boldsymbol{P}_{h})$, the agents play the induced finite-horizon incentivized dynamic game~\eqref{eq2}--\eqref{eq:total_cost} and choose the incentivized open-loop Nash equilibrium strategy. By Lemma~\ref{lem_ne}, the joint IONE strategy profile is given by  
\[
\boldsymbol{U}_{h}=\mu^\diamond(\xi_{h},\boldsymbol{P}_{h}).
\]
Let the state trajectory of dynamics~\eqref{eq2} under strategy $\boldsymbol{U}_{h}$ and initial state $\xi_h$ be $\boldsymbol{X}_{h}=[x_{h}(0)^\top,\ldots,x_{h}(T-1)^\top]^\top\in \mathbb{R}^{nT}$.

(iii) Noisy observation. At the end of each season, the planner observes the noisy state-action trajectory
\begin{equation}\label{eq:noisy_obs}
\boldsymbol{Z}_h
=
\begin{bmatrix}
   \boldsymbol{X}_{h}\\
   \boldsymbol{U}_{h}
\end{bmatrix}
+\boldsymbol{\epsilon}_{h},
\end{equation}
where $\boldsymbol{\epsilon}_h\in\mathbb{R}^{(n+mL)T}$ is an $\mathcal F_h$-measurable noise and satisfies
$\mathbb{E}[\boldsymbol{\epsilon}_h\mid \mathcal{F}_{h-1}]=0.$ 
Since observation \(\boldsymbol{Z}_h\) is \(\mathcal{F}_h\)-measurable, the planner can use it to select incentives in season \(h+1\).

We note that the incentive sequence $\{\boldsymbol{P}_h\}_{h\geq1}$ is \emph{adapted} to the filtration $\{\mathcal{F}_{h-1}\}_{h\geq1}$, i.e., $\boldsymbol{P}_h$ is $\mathcal{F}_{h-1}$-measurable for every $h\geq1$.  This is precisely what makes the resulting incentive design \emph{adaptive}.

For each season $h\geq 1$, given the initial state $\xi_h$, the \emph{socially optimal strategy} $\boldsymbol{U}^{*}_{h}$ is the solution to the social cost minimization problem, i.e.,
\begin{equation}\label{social_solution}
\boldsymbol{U}^{*}_{h}=\operatorname*{argmin}_{\boldsymbol{U}\in \mathbb{R}^{mLT}}\Phi(\xi_h, \boldsymbol{U}),\quad \text{subject to \eqref{eq2},}    
\end{equation}
where social cost $\Phi$ is defined in \eqref{social_cost}.
The optimal $\boldsymbol{U}^*_{h}$ is unique due to the strict convexity of the cost $\Phi$. The vector $\boldsymbol{U}^*_{h}$ can be decomposed as
 $\boldsymbol{U}^*_{h}
=
[
\boldsymbol{u}_{1,h}^{*\top},\ldots,\boldsymbol{u}_{L,h}^{*\top}]^{\top}$ where $\boldsymbol{u}_{i,h}^*
=
[
u_{i,h}^*(0)^{\top},\ldots,u_{i,h}^*(T-1)^{\top}]^{\top}
\in\mathbb R^{mT}.$ 
The corresponding state trajectory generated by the system dynamics \eqref{eq2} is denoted by
$\{x^*_h(t)\}_{t=0}^{T}$.
 Based on classical LQ optimal control results, the structure of the socially optimal strategy can be characterized explicitly, as detailed in Appendix~\ref{sec:social}.
\subsection{Problem Formulation}
In this paper, we consider the case where the agents' preference parameters $\{Q_i, R_i, x_i^d\}_{i=1}^L$ are unknown to the planner.
The planner adaptively designs the incentive profile $\{\boldsymbol P_h\}_{h\geq 1}$
using only observations from previous interactions, with the goal of steering the agents' induced IONE strategy profile \(\boldsymbol{U}_h\) toward the socially optimal strategy profile \(\boldsymbol{U}_h^*\). To evaluate the social coordination performance of the online adaptive incentive sequence $\{\boldsymbol{P}_h\}_{h=1}^{H}$, we define the average regret up to season $H$ as
\begin{equation}\label{regret}
\mathcal{R}(H)=\frac{1}{H}\sum_{h=1}^{H}
\left[
\Phi\left(\xi_h, \boldsymbol{U}_{h}\right)
-
\Phi\left(\xi_h,  \boldsymbol{U}^{*}_h\right)
\right].
\end{equation} 

Another performance metric in economic mechanism design
is to
promote {\it budget balance} (BB), i.e., to minimize the net surplus or deficit
resulting from the total monetary transfers among the participants (see, e.g., Definition~2.40 in~\cite{Narahari2009}).
We define the average net budget up to season $H$ as
\begin{equation}\label{eq:budget}
\mathcal{B}(H)=\frac{1}{H}\sum_{h=1}^{H}\sum_{i=1}^{L}
\sum_{t=0}^{T-1}
\gamma_i\bigl(x_h(t),u_{i,h}(t), t\bigr).
\end{equation}

Then, the main problem 
is summarized as follows.
\begin{problem}[AID for LQ Dynamic Games]\label{problem}
 Suppose the agents' preference parameters $\{Q_i, R_i, x_i^d\}_{i=1}^L$ are unknown to the
planner. Construct an online adaptive incentive design algorithm for the planner that, at each season $h\geq 1$, selects an $\mathcal{F}_{h-1}-$measurable incentive profile \(\boldsymbol{P}_h\) based only on the initial state \(\xi_h\) and the observation history \(\{\boldsymbol{Z}_\tau\}_{\tau=1}^{h-1}\),
such
that both the average regret and net budget converge to $0$ as $H \to \infty$, almost surely, i.e.,
\[\lim_{H\to \infty}\mathcal R(H)=0, \text{ and } \lim_{H\to \infty}\mathcal B(H)=0, \qquad a.s.\]
\end{problem}

Problem \ref{problem} solves optimal social coordination while achieving a vanishing average regulation regret and average net transfer. Next,  we introduce the main assumptions.

\begin{assumption}\label{assum2}
For each \(i\in[L]\), \((A,B_i)\) is controllable, and \(A\) is nonsingular. The horizon length $T\geq n+2$.
\end{assumption}
\begin{assumption}\label{assum_noise}
The noise sequence $\{\boldsymbol{\epsilon}_h, \mathcal{F}_h\}$ is a martingale difference sequence. 
Moreover, there exist constants $\beta>4$ and $\sigma>0$ such that 
\begin{equation*}
\begin{aligned}
\mathbb E\!\left[
    \boldsymbol{\epsilon}_h\boldsymbol{\epsilon}_h^\top
    \mid \mathcal F_{h-1}
\right]
= \sigma^2 I_{(n+mL)T},\;\; \text{a.s., $\quad$ and}
\end{aligned}
\end{equation*}
\begin{equation*}
\begin{aligned}
\sup_{h\geq 1}
\mathbb E\!\left[
    \|\boldsymbol{\epsilon}_h\|^\beta
    \mid \mathcal F_{h-1}
\right]
< \infty
\qquad \text{a.s.}
\end{aligned}
\end{equation*}

\end{assumption}
\begin{assumption}\label{assum1}
The initial state sequence $\{\xi_h\}_{h\geq 1}$ is bounded and $\xi_h$ is $\mathcal F_{h-1}$-measurable for each $h\geq 1$.
\end{assumption}

\begin{remark}
Assumption~\ref{assum1} allows both episodic and rolling-horizon settings. In the episodic setting, each $\xi_h$ can be independent of $\xi_{j}$, $\forall h\neq j$. In the rolling-horizon setting, the initial state of season $h+1$ equals the final state of season $h$, i.e., $\xi_{h+1}=x_h(T)$. A sufficient condition for Assumption~\ref{assum1} in the rolling-horizon setting is the stability of the corresponding closed-loop system. More specifically, if there exist constants $C>0$ and $\rho\in(0,1)$ such that, for all $h\geq1$, $\|x_h(T)\|
\leq
\rho \|\xi_h\|+C,$
then $\{\xi_h\}_{h\geq 1}$ is bounded. 
\end{remark}
\subsection{Incentive Design under Full Information}
Before investigating information asymmetry, where the agents' parameters $\{Q_i,R_i,x_i^d\}_{i=1}^L$ are unknown to the planner, we begin with the full information scenario in which these values are known.
With full information,  we show that $\lim_{H\to \infty}\mathcal{R}(H)=0$ and $\lim_{H\to \infty}\mathcal{B}(H)=0$ can be achieved by a family of incentive sequences
$\{\boldsymbol{P}_h\}_{h\geq 1}$, and the following two equalities hold simultaneously for every $h\geq 1$:
\begin{equation}\label{eq:equilibrium_social_cost}   \mu^\diamond(\xi_h,\boldsymbol{P}_h)=\boldsymbol{U}^*_h,
\end{equation}
 and
\begin{equation}\label{eq:min_payment_incentive}
\sum_{i=1}^{L}
\sum_{t=0}^{T-1}
\gamma_i\bigl(x_h(t),u_{i,h}(t), t\bigr)=0. 
\end{equation}

 The following proposition provides an explicit characterization of 
 the solutions to \eqref{eq:equilibrium_social_cost} and \eqref{eq:min_payment_incentive}. The proof can be found
 in Appendix~\ref{sec:proofs}.


\begin{proposition}\label{pro1}
Consider 
the socially optimal strategy profile defined in
\eqref{social_solution}. Suppose that Assumption~\ref{assum_ne} holds.
For each season $h\geq 1$, given private cost parameters
$\{Q_i,R_i,x_i^d\}_{i=1}^L$ and the initial state $\xi_h$, an incentive parameter profile $\boldsymbol{P}_h$ satisfies \eqref{eq:equilibrium_social_cost} and \eqref{eq:min_payment_incentive} if and only if
\begin{equation}\nonumber
\boldsymbol{P}_h
\in
-\mathcal H_h^{*\top}
\bigl(\mathcal H^*_h\mathcal H_h^{*\top}\bigr)^{-1}
b_h
+
\ker\bigl(\mathcal H^*_h\bigr),
\end{equation}
where 
\begin{equation}\label{eq:optimal_p}
\begin{aligned}
&\mathcal H^*_h=
\begin{bmatrix}
\mathcal H\\
\mathbf 1_{LT}^{\top}\mathcal G^*_h
\end{bmatrix},
\;
b_h=
\begin{bmatrix}
M \boldsymbol{U}^*_h+\eta_h\\
0
\end{bmatrix},
\end{aligned}
\end{equation}
with $M$ and $\mathcal H$ defined in Lemma~\ref{lem_ne}, and
\begin{equation}\label{eq:eta_label}
  \begin{gathered}
     \mathcal G^*_h
=
\operatorname{diag}\!\bigl\{
\mathcal G_{1,h}^*,\ldots,\mathcal G_{L,h}^*
\bigr\},\; \eta_h
=
\begin{bmatrix}
\eta_{1,h}^\top&
\cdots &
\eta_{L,h}^\top
\end{bmatrix}^{\top}, \\
\mathcal G_{i,h}^*
=
\operatorname{diag}\!\bigl\{
g_{i,h}^*(0),\ldots,g_{i,h}^*(T-1)
\bigr\},\\
g_{i,h}^*(t)
=
\begin{bmatrix}
x_h^*(t)^\top, &
u_{i,h}^*(t)^\top, &
1
\end{bmatrix},\\
\eta_{i,h}
=
\Gamma_i^\top
[
(I_T\otimes Q_i)\bar A\xi_h
-\mathbf 1_T\otimes(Q_i x_i^d)].
   \end{gathered}
\end{equation}
\end{proposition}

Proposition \ref{pro1} answers how one designs an incentive sequence
incentive $\{\boldsymbol{P}_h\}_{h\geq 1}$ for zero regret and zero net budget,
when the planner knows the agents' preference parameters $\{Q_i,R_i,x_i^d\}_{i=1}^L$.
This result serves as a benchmark and provides the basis for the subsequent adaptive incentive design under
unknown preference parameters.

To formulate the associated estimation problem, let
$\bar x_i^d=Q_i x_i^d$ and
$d_\theta=\frac{m(m+1)}{2}+\frac{n(n+1)}{2}+n$,
and define the parameter vectors
\[
\theta_i
=
\begin{bmatrix}
\operatorname{vech}(R_i)^\top &
\operatorname{vech}(Q_i)^\top &
(\bar x_i^d)^\top
\end{bmatrix}^{\top}
\in \mathbb R^{d_\theta},
\]
\[
\theta
=
\begin{bmatrix}
\theta_1^\top & \cdots & \theta_L^\top
\end{bmatrix}^{\top}
\in \mathbb R^{L d_\theta}.
\] 

\section{Adaptive Preference Learning}\label{sec4}
In this section, we study the online estimation of the
preference parameter vector $\theta$ from noisy observations
of the agents' strategic behavior. 

We first use Pontryagin's
maximum principle to express the agents' equilibrium conditions
as a linear regression in $\theta$. To this end, define the
selection matrices
\begin{equation*}
E_x=
\begin{bmatrix}
I_{nT} & 0_{nT\times LmT}
\end{bmatrix},\quad \text{and}    
\end{equation*}
\[
E_{u_i}
=
\begin{bmatrix}
0_{mT\times (nT+(i-1)mT)}
&
I_{mT}
&
0_{mT\times (L-i)mT}
\end{bmatrix}.
\] 
Since only noisy observations $\boldsymbol{Z}_h$ are available, the resulting
regression relation contains a noise-dependent perturbation, as
made explicit in the following lemma.

\begin{lemma}\label{lem31}
Suppose that Assumption~\ref{assum_ne} holds. For each season \(h\geq 1\) and
each agent \(i\in[L]\), define
\begin{equation}\label{phi}
\begin{aligned}
\phi_{i,h}
&=
\Bigl[
\mathcal{K}_{m,T}\!\left(E_{u_i}\boldsymbol{Z}_h\right)\mathscr{D}_m,
\\[-1mm]
&\qquad \qquad
\Gamma_i^{\top}
\mathcal{K}_{n,T}\!\left(E_x\boldsymbol{Z}_h\right)\mathscr{D}_n,
\;\;
-\Gamma_i^{\top}(\mathbf{1}_T\otimes I_n)
\Bigr],
\\[1mm]
\delta_{i,h}
&=
\Bigl[
\mathcal{K}_{m,T}\!\left(E_{u_i}\boldsymbol{\epsilon}_h\right)
\mathscr{D}_m,\;\; \Gamma_i^{\top}
\mathcal{K}_{n,T}\!\left(E_x\boldsymbol{\epsilon}_h\right)
\mathscr{D}_n,
\;\;
0
\Bigr],\\[1mm]
y_{i,h}
&=
-\frac{1}{2}
\left(\Gamma_i^{\top}S_x+S_u\right)\boldsymbol{p}_{i,h},
\end{aligned}
\end{equation}
where the operator
\(\mathcal K_{\cdot,\cdot}(\cdot)\) and \(\Gamma_i\) are defined in \eqref{eq:operator} and \eqref{eq:gamma}, respectively, and 
$\mathscr{D}_m\in\mathbb R^{m^2\times m(m+1)/2}$, $\mathscr{D}_n\in\mathbb R^{n^2\times n(n+1)/2}$ denote the duplication matrices. 
Then the agents' state and
IONE trajectories satisfy
\begin{equation}\label{regression}
(\phi_{h}-\delta_h)\theta=y_h,   
\end{equation}
where $\phi_h=\operatorname{diag}\{\phi_{1,h},\cdots \phi_{L,h}\},\  \delta_h=\operatorname{diag}\{\delta_{1,h},\cdots \delta_{L,h}\}$, and $ y_h=[y_{1,h}^{\top},\cdots y_{L,h}^{\top}]^{\top}.
$ 
\end{lemma}

Lemma~\ref{lem31} shows that the regression vector in the perturbed linear regression model \eqref{regression} is corrupted by measurement noise. Therefore, the canonical techniques from the least-squares regression lead to a biased preference estimate. The
following result proposes a bias-corrected loss whose conditional expectation is
minimized by the true parameter.
\begin{remark}\label{cor1}
For each season
\(h\geq 1\) and any  $\mathcal{F}_{h-1}-$measurable \(\boldsymbol{P}_h\), the true cost parameter $\theta$
satisfies 
    \begin{equation}\label{eq:loss}
    \theta\in
    \operatorname*{argmin}_{\vartheta\in \mathbb{R}^{d_\theta L}}\mathbb{E}\left[\|y_h-\phi_h \vartheta\|^2-\vartheta^{\top}\Xi \vartheta\mid \mathcal{F}_{h-1}\right],
    \end{equation}
    where
    \begin{equation}\label{eq:xi}
        \Xi=\mathbb{E}[\delta_h^{\top}\delta_h\mid\mathcal{F}_{h-1}].
    \end{equation}
Specifically, by the explicit expression of $\delta_h$ in \eqref{phi} and Assumption~\ref{assum_noise},
$\Xi$ can be computed explicitly as
$\Xi=\operatorname{diag}\{\Xi_{1},\cdots \Xi_{L}\}$, and, for each $i=1,\ldots,L$,
{\small
\begin{equation}\label{xi}
\begin{aligned}
\Xi_i=\sigma^2
\begin{bmatrix}
T\mathscr{D}_m^{\top}\mathscr{D}_m & 0 & 0\\
0 & \mathscr{D}_n^{\top}(I_n \otimes \operatorname{btr}(\Gamma_i\Gamma_i^\top))\mathscr{D}_n & 0\\
0&0&0_{n\times n}
\end{bmatrix}.
\end{aligned}
\end{equation}}
Here, partition $\Gamma_i\in\mathbb R^{nT\times mT}$ into
$T\times T$ blocks as $\Gamma_i=[\Gamma_{i,ts}]_{t,s=1}^{T}$,
where $\Gamma_{i,ts}\in\mathbb R^{n\times m}$, and define $\operatorname{btr}(\Gamma_i\Gamma_i^\top)
=
\sum_{t=1}^{T}\sum_{s=1}^{T}
\Gamma_{i,ts}\Gamma_{i,ts}^\top
\in\mathbb R^{n\times n}.$
\end{remark}

The regularization term in the loss \eqref{eq:loss} is introduced to compensate for the errors-in-variables bias caused by the observation noise in the regression vector.
According to Remark \ref{cor1}, we construct the following modified recursive LS algorithm to estimate $\theta$:
\begin{equation}\label{alg}
 \begin{aligned}
&\hat{\theta}_{h+1}=\hat{\theta}_{h}+\Lambda_{h+1}^{\dagger}\phi_h^{\top}(y_h-\phi_h\hat{\theta}_{h})+\Lambda_{h+1}^{\dagger}\Xi\hat{\theta}_{h},\\
& \Lambda_{h+1}=\Lambda_{h}+\phi_h^\top\phi_h-\Xi,  
 \end{aligned}
 \end{equation}
where $\hat{\theta}_{h+1}$ is the estimate update of $\theta$ at
	the end of season $h$; $\phi_{h}, y_h,$ and $\Xi$  are defined as in \eqref{phi} and \eqref{xi}; $\Lambda_h^{\dagger}$ denotes the Moore-Penrose pseudoinverse of $\Lambda_h$; the initial value $\hat{\theta}_{1}\in \mathbb{R}^{d_\theta L}$ and $\Lambda_1\succ0$  can be chosen arbitrarily.


We next establish the consistency of the recursive estimator \eqref{alg}. The result gives
an almost-sure estimation error bound under an excitation condition on the
information matrix \(\Lambda_h\), which allows us to avoid the assumption of the Persistent Excitation (PE) condition.
\begin{theorem}\label{thm1}
Suppose that Assumptions~\ref{assum_ne}--\ref{assum1} hold, and consider
the noisy regression model \eqref{regression}. Let $\{\boldsymbol{P}_h\}_{h\ge 1}$ be any
incentive sequence such that $\boldsymbol{P}_h$ is $\mathcal{F}_{h-1}$-measurable. If
\begin{equation}\label{growth_phi}
\|\phi_h\|=o\bigl(h^{\alpha}\bigr),\; a.s.,\ h\to\infty,
\end{equation}
for some $\alpha\in(0,\tfrac{1}{2})$, and as $h\to\infty,$
\begin{equation}\label{excitation}
\bigl(\lambda_{\max}(\Lambda_h)\,\log\log\lambda_{\max}(\Lambda_h)\bigr)^{1/2}
=o\bigl(\lambda_{\min}(\Lambda_h)\bigr),\; a.s.,
\end{equation}
then the recursive estimate $\hat{\theta}_h$ generated by \eqref{alg} is
strongly consistent, and satisfies as $h\to\infty,$
\begin{equation}\label{thm1_convergence}
\left\|\theta-\hat{\theta}_h\right\|^2=O\left(\frac{\lambda_{\max}\{\Lambda_h\}\log \log (\lambda_{\max}\{\Lambda_h\})}{(\lambda_{\min}\{\Lambda_h\})^2}\right),\ a.s.    
\end{equation}
\end{theorem}
\begin{remark}
Theorem~\ref{thm1} establishes the strong consistency of the proposed estimators under the excitation condition~\eqref{excitation} on the trajectory \(\{\boldsymbol{Z}_h\}_{h\geq 1}\). Compared to the identifiability conditions in inverse optimal control and games
\cite{inverse_control, zhang2021discrete, QU20232134, inverse_game_Observations,Inverse_Reinforcement}, which typically require the trajectory 
\(\{\boldsymbol{Z}_h\}_{h\geq 1}\) to be i.i.d., condition \eqref{excitation} allows temporal dependence and nonstationarity, provided sufficient information accumulates in all parameter directions. This accommodates the temporal dependence induced by the learning-and-decision feedback loop considered here, where past observations determine the current parameter estimates, which guide the design of incentives that, in turn, shape future observations.
\end{remark}

Theorem~\ref{thm1} specifies the excitation required for consistent preference learning but does not prescribe how to generate it. The next section addresses this active-learning task by designing incentives that produce sufficiently informative trajectories while achieving sublinear social-cost regret.

\section{Adaptive Incentive Design}\label{sec5}
In this section, we develop an adaptive incentive design algorithm based on the
online preference learning algorithm proposed in Section~\ref{sec4}.  The resulting Algorithm 1 solves Problem 1 and gives explicit rates for the almost-sure convergence of the average regret and the average net-budget imbalance to zero.

We first recover the component estimates of $R_i$, $Q_i$, and $\bar x_i^d$ from $\hat\theta_h$.
Write $\hat\theta_h
=
[\hat\theta_{1,h}^{\top},\ldots,\hat\theta_{L,h}^{\top}]^{\top}$ where
$\hat\theta_{i,h}\in\mathbb R^{d_\theta}.$ For each \(i\in[L]\), partition $\hat\theta_{i,h}\in\mathbb{R}^{d_\theta}$ as
$
\hat\theta_{i,h}=\big[(\hat\theta_{i,h}^{(1)})^{\top},(\hat\theta_{i,h}^{(2)})^{\top},(\hat\theta_{i,h}^{(3)})^{\top}\big]^{\top},
$
where $\hat\theta_{i,h}^{(1)}\in\mathbb{R}^{m(m+1)/2}$, $\hat\theta_{i,h}^{(2)}\in\mathbb{R}^{n(n+1)/2}$, and $\hat\theta_{i,h}^{(3)}\in\mathbb{R}^n$. The corresponding estimates of the cost matrices $R_i$, $Q_i$, and the target state $\bar x_i^d$ at season $h$ are denoted by $\hat R_{i,h}$, $\hat Q_{i,h}$, and $\hat {\bar x}_{i,h}^d$, respectively, and satisfy
\[
\operatorname{vech}(\hat R_{i,h})=\hat\theta_{i,h}^{(1)},
\quad
\operatorname{vech}(\hat Q_{i,h})=\hat\theta_{i,h}^{(2)},
\quad
\hat {\bar x}_{i,h}^d=\hat\theta_{i,h}^{(3)}.
\]

Given the estimates \((\hat R_{i,h},\hat Q_{i,h}, \hat {\bar x}_{i,h}^d)\) at
season \(h\), we construct the adaptive incentive according to Proposition~\ref{pro1} following the certainty-equivalence principle, by replacing the unknown cost parameter $\theta$ with its latest estimates, i.e., let
\begin{equation}\label{ce}
\bar{\boldsymbol{P}}_h=-\mathcal H_h^{*\top}(\mathcal H^*_h\mathcal H_h^{*\top})^{-1}\ \hat{b}_h,
\end{equation}
where $\mathcal{H}_h^{*}$ is given in \eqref{eq:optimal_p}, and 
\[
\hat{b}_h
=
\begin{bmatrix}
\hat{M}_h\boldsymbol{U}^*_h+\hat{\eta}_h\\
0
\end{bmatrix}.
\]
Here, $\hat{M}_h$ and $\hat{\eta}_h$ are obtained from $M$ and $\eta_h$ in \eqref{uniquness} and \eqref{eq:eta_label}, respectively, by replacing the unknown preference parameters $\{Q_i,R_i,x_i^d\}_{i=1}^{L}$ with their estimates, i.e.,
\[
\begin{gathered}
\hat{\eta}_h =
\begin{bmatrix}
\hat{\eta}_{1,h}^{\top} & \cdots & \hat{\eta}_{L,h}^{\top}
\end{bmatrix}^{\top},
\\[3pt]
\hat{\eta}_{i,h} =
\Gamma_i^{\top}
\Bigl[
\hat{\bar Q}_{i,h}\bar A\,\xi_h
-\mathbf{1}_T\otimes\hat{\bar x}_{i,h}^{d}
\Bigr],
\quad i\in[L],
\end{gathered}
\]
and
\begin{equation*}
\begingroup
\setlength{\arraycolsep}{2pt}
\small
\hat M_h
=
\begin{bmatrix}
\Gamma_1^\top \hat{\bar Q}_{1,h}\Gamma_1
+\hat{\bar R}_{1,h}
&
\cdots
&
\Gamma_1^\top \hat{\bar Q}_{1,h}\Gamma_L
\\
\vdots
&
\ddots
&
\vdots
\\
\Gamma_L^\top \hat{\bar Q}_{L,h}\Gamma_1
&
\cdots
&
\Gamma_L^\top \hat{\bar Q}_{L,h}\Gamma_L
+\hat{\bar R}_{L,h}
\end{bmatrix}
\endgroup
\end{equation*}
where $\hat{\bar Q}_{i,h}
=
I_T\otimes\hat Q_{i,h},$ $\hat{\bar R}_{i,h}
=
I_T\otimes\hat R_{i,h}$, for $i\in[L]$.

To ensure boundedness of the designed incentive, we project the
certainty-equivalence output $\bar{\boldsymbol{P}}_h$ given by \eqref{ce} onto a slowly
growing ball
\begin{equation}\label{certain}
\hat{\boldsymbol{P}}_h=\Pi_{\mathcal{P}_h}(\bar{\boldsymbol{P}}_h),
\end{equation}
where $\mathcal P_h
=
\{\boldsymbol{P}\in\mathbb R^{(m+n+1)LT}:\|\boldsymbol{P}\|\le \log (h+1)\}$. The logarithmic radius is loose enough
that the projection is inactive once $\bar{\boldsymbol{P}}_h$ stabilizes, yet tight
enough to control the influence of inaccurate early-season estimates.

The projected certainty-equivalent incentive $\hat{\boldsymbol{P}}_h$ uses the current preference estimate to coordinate agents, but does not by itself guarantee informative data for subsequent learning. To actively generate such information, we perturb the incentive
\begin{equation}\label{adaptive_incentive}
\begin{aligned}
\boldsymbol{P}_h=\hat{\boldsymbol{P}}_h+\gamma_hv_h, 
\end{aligned}
\end{equation}
where $\gamma_h=h^{-\gamma}$ is a shrinking probing scale, and the rate $\gamma \in [0,\frac{1}{4n})$. 
The sequence \(\{v_h\}_{h\geq 1}\)  is an i.i.d. sequence, and for each $h\geq1$, $v_h$ is independent of $\sigma\!\left(
\xi_1,\ldots,\xi_h,\,
\boldsymbol{\epsilon}_1,\ldots,\boldsymbol{\epsilon}_{h-1},\,
v_1,\ldots,v_{h-1}
\right).$ Moreover, each \(v_h\in\mathbb R^{(m+n+1)LT}\) satisfies for some $\bar{v}>0$, \footnote{For definiteness in the present setting,  for $h\geq 0$, one may take
\[
\mathcal F_h
=
\sigma\left(\{\xi_j:1\leq j\leq h+1 \}\cup
\{\epsilon_j:1\leq j\leq h\}
\cup
\{v_j:1\leq j\leq h+1\}
\right)
\]}
\begin{equation}\label{probing}
\mathbb E[v_h]=0,\quad
   \mathbb{E}[v_hv_h^{\top}]=I,\ \ \|v_h\|\leq \bar{v},\quad a.s.,\ \ \forall h\geq 1. 
\end{equation}
The two terms in \eqref{adaptive_incentive} serve complementary purposes: $\hat{\boldsymbol{P}}_h$ exploits the current estimate for social coordination, whereas $\gamma_hv_h$ probes agents’ responses to support preference learning. Because the planner chooses payments rather than agents’ actions, this probing must generate excitation indirectly through the open-loop Nash response and the system dynamics. 

The complete adaptive incentive design algorithm for the LQ dynamic games is summarized in
Algorithm~\ref{alg:adaptive_incentive}.
\begin{algorithm}[htbp]
\caption{Adaptive incentive design for LQ dynamic games (LQ-AID)}
\label{alg:adaptive_incentive}
\begin{algorithmic}[1]
\STATE \textbf{Initialization:} Choose initial parameters $ \hat{\theta}_1$ and $\Lambda_1$ such that $\hat\theta_1\in \mathbb{R}^{d_{\theta}L}$ and
$\Lambda_1\succ0.$

\FOR{\(h=1,2,\ldots\)}
    \STATE The planner observes the initial state \(\xi_h\), and designs the incentive based on the current estimate \(\hat{\theta}_{h}\):
    \[
\boldsymbol{P}_h=\hat{\boldsymbol{P}}_h+\gamma_hv_h,
    \qquad
    \gamma_h=h^{-\gamma},
    \]
    where $\hat{\boldsymbol{P}}_h$ is given in \eqref{certain}, and $v_h$ satisfies \eqref{probing}. 
    \STATE The agents observe \(\boldsymbol{P}_h\) and play the open-loop Nash equilibrium
    strategy profile \(\boldsymbol{U}_h=\mu^\diamond(\xi_h, \boldsymbol{P}_h)\), generating the IONE response
    trajectory \((\boldsymbol{X}_h,\boldsymbol{U}_h)\).
    \STATE The planner observes the noisy trajectory
    \(\boldsymbol{Z}_h
    \), and updates
\begin{equation}\nonumber
 \begin{aligned}
&\hat{\theta}_{h+1}=\hat{\theta}_{h}+\Lambda_{h+1}^{\dagger}\phi_h^{\top}(y_h-\phi_h\hat{\theta}_{h})+\Lambda_{h+1}^{\dagger}\Xi\hat{\theta}_{h},\\
& \Lambda_{h+1}=\Lambda_{h}+\phi_h^\top\phi_h-\Xi,  
 \end{aligned}
 \end{equation}
 where $\phi_h$, $y_h$, and $\Xi$ are defined in \eqref{phi} and \eqref{xi}. 
\ENDFOR
\end{algorithmic}
\end{algorithm}

We next analyze the convergence rate of the parameter estimates generated by
Algorithm~\ref{alg:adaptive_incentive}. By Theorem~\ref{thm1}, this relies on the condition \eqref{excitation} when the
closed-loop trajectories are generated by the designed incentives given in the algorithm. To this end, the next lemma
shows that the designed adaptive incentive $\{\boldsymbol{P}_h\}_{h\geq 1}$ provides the
required growth of \(\lambda_{\min}(\Lambda_h)\).
\begin{lemma}\label{lem:adaptive_excitation}
Suppose that Assumptions~\ref{assum_ne}--\ref{assum1} hold. In
Algorithm~\ref{alg:adaptive_incentive}, there exist an almost surely finite
random index \(h_0\) and a positive random constant \(c_\Lambda\) such that
\[
\lambda_{\min}(\Lambda_h)
\geq
c_\Lambda h^{1-2n\gamma},
\qquad
a.s.,
\qquad
\forall h\geq h_0 .
\]
\end{lemma}

Lemma~\ref{lem:adaptive_excitation} quantifies the information generated by active incentive probing. Despite the vanishing perturbation amplitude, the smallest eigenvalue of the information matrix grows fast enough to satisfy the learning condition in Theorem~\ref{thm1}. This yields the following convergence rate of preference estimation.


\begin{theorem}\label{thm2}
Suppose that Assumptions \ref{assum_ne}--\ref{assum1} hold. Then the parameter estimates \(\hat{\theta}_h\) generated by Algorithm~1 satisfy the following upper bound:
   \[
   \|\hat \theta_h-\theta\|=O\left(\frac{\sqrt{h\log\log h}}{h^{1-2n\gamma}}\right),\ \ a.s.,\ \ h\to \infty,
   \]
   where $\gamma\in [0,\frac{1}{4n})$ is given in \eqref{adaptive_incentive}.
\end{theorem}

The detailed proof of Theorem 2 is deferred  to the next section, and we now present the sublinear growth rate of the cumulative regret and budget for the proposed adaptive incentive algorithm. 
\begin{theorem}\label{thm3}
Suppose that Assumptions~\ref{assum_ne}--\ref{assum1} hold. Then the average regret $\mathcal R(H)$ and the average net budget
$\mathcal B(H)$ satisfy that
\[
\mathcal R(H)
=
O\left(
(\log\log H)H^{-1+4n\gamma}
+
H^{-2\gamma}
\right),\ \ a.s.,
\]
and
\[
\mathcal B(H)
=
O\left(
(\log\log H)^{\frac12}H^{-\frac12+2n\gamma}
+
H^{-\gamma}
\right),\ \ a.s., 
\]
as $H\to\infty.$
Here, $\mathcal R(H)$ and $\mathcal B(H)$ are defined in
\eqref{regret} and \eqref{eq:budget}, respectively. 
\end{theorem}

Theorem~\ref{thm3} shows that the performance loss of the adaptive incentive
has two sources. The first term in both $\mathcal R(H)$ and $\mathcal B(H)$ is actually caused by the parameter learning error
\(\|\hat\theta_h-\theta\|\), while the second term in both bounds is caused by the exploration
signal \(\gamma_hv_h\).
We note that the dependence of the first term on the state dimension $n$ originates from the excitation bound in Lemma~\ref{lem:adaptive_excitation}, where the growth rate of the smallest eigenvalue of the information matrix $\Lambda_h$ depends explicitly on $n$ and, in turn, determines the convergence rate of the parameter estimates. Consequently, sufficient excitation in higher-dimensional systems requires a more slowly decaying exploration signal, which leads to the condition $0<\gamma<1/(4n)$ for sublinear regret. Balancing the polynomial orders of the two terms yields
$\gamma^\star=1/(4n+2)$.
\section{Proof of main results}\label{sec7}
In this section, we provide proofs of Theorems \ref{thm1}-\ref{thm3}.

\subsection{Proof of Theorem~\ref{thm1}}
To prove Theorem~\ref{thm1}, we first establish several preliminary bounds in the following lemma.
\begin{lemma}\label{lem:preliminary_bounds}
Under the conditions  of Theorem~\ref{thm1}, for each $h\geq 1$, define
\begin{equation}\label{barphi}
\bar\phi_h=\phi_h-\delta_h,
\qquad
g_h=\sum_{j=1}^h\|\bar\phi_j\|^2.
\end{equation}
Then as $h\to\infty$,
\begin{align}
&h
=
O\!\left(g_h\right),\; g_h
=
O\!\left(\lambda_{\max}(\Lambda_{h+1})\right), \quad a.s.,
\label{eq:pre_g}
\end{align}
and, for some $\bar\alpha\in
\left(
\max\left\{\alpha,\frac1\beta\right\},
\frac12
\right)$,
\begin{equation}\label{eq:pre_growth}
\|\bar\phi_h\|
=
o(h^{\bar\alpha}),\quad a.s.    \end{equation}
\end{lemma}
The proof of Lemma~\ref{lem:preliminary_bounds} is given in Appendix~\ref{sec:proofs}. 

{\bfseries\itshape Proof of Theorem~\ref{thm1}:} For each $h\geq 1$, define
\begin{equation}\nonumber
    \begin{aligned}
\tilde{\theta}_h=\theta-\hat{\theta}_h.
    \end{aligned}
\end{equation}
By~\eqref{alg} and \eqref{regression}, we have
\begin{equation}\nonumber
\begin{aligned}
&\tilde{\theta}_{h+1}=(I+\Lambda_{h+1}^{\dagger}\Xi)\tilde{\theta}_{h}-\Lambda_{h+1}^{\dagger}\phi_h^{\top}(y_h-\phi_h\hat{\theta}_{h})-\Lambda_{h+1}^{\dagger}\Xi\theta\\
&=(I+\Lambda_{h+1}^{\dagger}\Xi-\Lambda_{h+1}^{\dagger}\phi_h^{\top}\phi_h)\tilde{\theta}_{h}+\Lambda_{h+1}^{\dagger}\phi_h^{\top}\delta_h\theta-\Lambda_{h+1}^{\dagger}\Xi\theta.
\end{aligned}
\end{equation}
By \eqref{eq:pre_g} in Lemma \ref{lem:preliminary_bounds}, 
\begin{equation}\label{eq:lambda_max}
h=O\left(\lambda_{\max}(\Lambda_h)\right),\;\; a.s., \;\; h\to \infty.    
\end{equation}
Together with \eqref{excitation}, this implies that 
$\lambda_{\min}(\Lambda_h)\to\infty$ as $h\to \infty$ almost surely. Therefore, there
exists an almost surely finite $h_0$ such that $\Lambda_h$ is
invertible for all $h\ge h_0$.
Hence, for all $h\geq h_0$, $\Lambda_h^{\dagger}=\Lambda_h^{-1}$, and
\begin{equation}\label{eq:recursion}
\begin{aligned}
\Lambda_{h+1}\tilde{\theta}_{h+1}=&(\Lambda_{h+1}+\Xi-\phi_h^{\top}\phi_h)\tilde{\theta}_{h}+
\phi_h^{\top}\delta_h\theta-\Xi\theta\\
=&\;\Lambda_{h}\tilde{\theta}_{h}+
\phi_h^{\top}\delta_h\theta-\Xi\theta,
\end{aligned}
\end{equation}
where the second equality follows from the recursion in \eqref{alg}.
Iterating \eqref{eq:recursion} and using
$\phi_h=\bar{\phi}_h+\delta_h$ yields
\begin{equation}\label{eqq:lambda34}
\Lambda_{h+1}\tilde{\theta}_{h+1}=\Lambda_{h_0}\tilde{\theta}_{h_0}+\sum_{j=h_0}^{h}
\bar{\phi}_j^{\top}\delta_j\theta+\sum_{j=h_0}^{h}(\delta_j^{\top}\delta_j-\Xi)\theta.
\end{equation}

We now bound the two summation terms on the right-hand side (RHS) of~\eqref{eqq:lambda34} separately.
 For the last term, by \eqref{eq:xi}, for each $h\geq 1$, we have 
\[
\mathbb{E}[\delta_h^{\top}\delta_h-\Xi\mid \mathcal{F}_{h-1}]=0, \quad \text{a.s.}
\]
Hence, $\{\delta_h^{\top}\delta_h-\Xi,\,\mathcal{F}_h\}_{h}$ is a martingale difference sequence. Moreover,
\begin{equation}\nonumber
\begin{aligned}
&\mathbb{E}\big[\|\delta_h^{\top}\delta_h-\Xi\|^{\beta/2}\mid \mathcal{F}_{h-1}\big]=O\big(\mathbb{E}[\|\delta_h\|^{\beta}\mid \mathcal{F}_{h-1}]\big)+O(1)\\
=&\;O\big(\mathbb{E}[\|\epsilon_h\|^{\beta}\mid \mathcal{F}_{h-1}]\big)+O(1)=O(1), \;\; \text{a.s.},\;\; h\to \infty.
\end{aligned}
\end{equation}
Since $\beta>4$, applying the martingale convergence result in Lemma~\ref{lem:martingale} gives
\begin{equation}\label{eq:eta}
\left\|\sum_{j=h_0}^{h}(\delta_j^{\top}\delta_j-\Xi)\theta\right\|=O\big(\sqrt{h\log\log h}\big), \quad \text{a.s.}
\end{equation}

For the cross term, $\bar\phi_h$ is $\mathcal F_{h-1}$-measurable and
$\{\delta_h,\mathcal F_h\}_{h}$ is a martingale difference
sequence. Besides, Lemma \ref{lem:preliminary_bounds} gives $\|\bar{\phi}_h\|=o(h^{\bar \alpha})=o(g_h^{\bar \alpha})$. Therefore, by Lemma~\ref{lem:martingale},
\begin{equation}\label{eq:delta}
\left\|
\sum_{j=h_0}^{h}\bar\phi_j^\top\delta_j
\right\|
=
O\!\left(
\sqrt{g_h\log\log g_h}
\right),
\quad\text{a.s.}
\end{equation}
Hence, by \eqref{eq:pre_g},
\[
\left\|
\sum_{j=h_0}^{h}\bar\phi_j^\top\delta_j
\right\|
=
O\!\left(
\sqrt{
\lambda_{\max}(\Lambda_{h+1})
\log\log\lambda_{\max}(\Lambda_{h+1})
}
\right),
\;\text{a.s.}
\]
Finally, by~\eqref{eq:lambda_max}, \eqref{eqq:lambda34}, \eqref{eq:eta},  and the bound above, we obtain
\[
\|\Lambda_{h}\tilde{\theta}_h\|=O\left(\sqrt{\lambda_{\max}(\Lambda_h)\log\log \lambda_{\max}(\Lambda_h)}\right), \quad \text{a.s.}
\]
Since $\|\Lambda_{h}\tilde{\theta}_h\|\geq \lambda_{\min}(\Lambda_{h})\|\tilde{\theta}_h\|$, it follows that
\[
\|\tilde{\theta}_h\|=O\left(\frac{\sqrt{\lambda_{\max}(\Lambda_h)\log\log \lambda_{\max}(\Lambda_h)}}{\lambda_{\min}(\Lambda_{h})}\right), \quad \text{a.s.}
\]
\hfill $\square$
\subsection{Proof of Theorem~\ref{thm2}}
We first establish three lemmas that will be used to derive
the excitation properties of the information matrix $\Lambda_h$. The
first provides a lower bound for a structured block matrix.
\begin{lemma}\label{lem4}
Let $r=n+m+1$. For each $h\geq r$, let
\[
W_h
=
\begin{bmatrix}
w_{11}I_m & \cdots & w_{1r}I_m\\
\vdots & \ddots & \vdots\\
w_{h1}I_m & \cdots & w_{hr}I_m
\end{bmatrix}
\in\mathbb R^{hm\times rm},
\]
For each entry $w_{ij}\in \mathbb{R}$, let $\mathfrak{W}_{ij}\in\mathbb R^{mn\times mn}$ be an upper
triangular matrix of the form
\[
\mathfrak{W}_{ij}
=
\begin{bmatrix}
w_{ij}I_m & \tilde{w}_{ij}^{(1)}I_m & \cdots & \tilde{w}_{ij}^{(n-1)}I_m\\
0 & w_{ij}I_m & \cdots & \tilde{w}_{ij}^{(n-2)}I_m\\
\vdots & \vdots & \ddots & \vdots\\
0 & 0 & \cdots & w_{ij}I_m
\end{bmatrix},
\]
where $\tilde{w}_{ij}^{(k)}\in \mathbb{R}$.
Define the block matrix
\[
\mathscr W_h
=
\begin{bmatrix}
\mathfrak{W}_{11} & \cdots & \mathfrak{W}_{1r}\\
\vdots & \ddots & \vdots\\
\mathfrak{W}_{h1} & \cdots & \mathfrak{W}_{hr}
\end{bmatrix}
\in\mathbb R^{hmn\times rmn}.
\]
Suppose that there exist positive sequences
$\{\rho_h\}_{h\geq1}$ and $\{C_h\}_{h\geq1}$ such that, for every
$h\geq r$,
\begin{equation}\label{eq:d_h}
\lambda_{\min}(W_h^\top W_h)\geq \rho_h,
\end{equation}
and $|\tilde{w}_{ij}^{(k)}|\leq C_h$, $|w_{ij}|\leq C_h$
for all admissible $i,j,k$. Then there exists a constant
$\mu>0$, independent of $h$, such that
\[
\lambda_{\min}
\left(
\mathscr W_h^\top\mathscr W_h
\right)
\geq
\mu
\frac{h^{1-n}\rho_h^n}{C_h^{2n-2}}.
\]
\end{lemma}

The proof of Lemma~\ref{lem4} is deferred to
Appendix~\ref{sec:proofs}. For
$i\in[L]$, $h\geq1$, and $t=1,\ldots,T$, define augmented action-state vectors
\begin{equation}\label{eq:def_zeta}
\zeta_{i,h}(t)=[ u_{i,h}(t-1)^{\top},x_h(t)^{\top},   -1]^{\top}. 
\end{equation}
The following lemma shows how excitation of the \(\{\zeta_{i,h}(t)\}\) translates into excitation of the regression vectors \(\{\bar{\phi}_h\}\).
\begin{lemma}\label{lem:phi_excitation}
Suppose that Assumption \ref{assum2} holds. Assume that there exist positive sequences
$\{\rho_h\}_{h\ge1}$ and $\{C_h\}_{h\ge1}$,
with $\rho_h\to\infty$ as $h\to\infty$,
such that, for every $i\in[L]$ and $h\ge1$,
\begin{equation}\label{zeta}
\lambda_{\min}\left(
I_{m+n+1}+\sum_{\tau=1}^h \zeta_{i,\tau}(T-1) \zeta_{i,\tau}(T-1)^{\top}
\right)\ge \rho_h,
\end{equation}
and
\begin{equation}\label{trajectory_bound}
\max_{\tau\le h}
\sum_{t=1}^{T}
\|\zeta_{i,\tau}(t)\|
\le C_h .    
\end{equation}
Then there exists a constant $\mu>0$, independent of $h$,
such that, for every $h\ge1$, 
\begin{equation}\label{eq: lambda_min_bound}
\lambda_{\min}\left(\Lambda_1+
\sum_{\tau=1}^h
\bar\phi_{\tau}^{\top}\bar\phi_{\tau}
\right)
\ge
\mu\frac{h^{1-n}\rho_h^{n}}{C_h^{2n-2}},
\end{equation}
where $\{\bar\phi_{\tau}\}_{\tau\geq 1}$ is defined in \eqref{barphi}.
\end{lemma}

The proof of Lemma~\ref{lem:phi_excitation} is given in Appendix~\ref{sec:proofs}. 
The next lemma establishes the lower bound for augmented action-state vectors
under the exploration mechanism of Algorithm~1.
\begin{lemma}\label{lem4.2}
Consider Algorithm~1 under
Assumptions~\ref{assum_ne}--\ref{assum1}. Suppose that there exists
a positive sequence $\{ C_h\}_{h\ge1}$ with
$C_h=O(\log h)$ such that, for every $i\in[L]$,
\begin{equation}\label{eq:trajectory_bound_aux}
\max_{\tau\le h}
\sum_{t=1}^{T}
\|\zeta_{i,\tau}(t)\|
\le C_h,
\qquad\text{a.s.}
\end{equation}
Then there exists a constant $\mu>0$ such
that, for each $i\in[L]$, 
{\small
\[
\lambda_{\min}\!\left(I_{m+n+1}+
\sum_{\tau=1}^{h}
\zeta_{i,\tau}(T-1)
\zeta_{i,\tau}(T-1)^\top
\right)
\ge
\mu\frac{h^{1-2\gamma}}{1+ C_h^2},\; a.s.
\]}
\end{lemma}

The proof of Lemma~\ref{lem4.2} is given in Appendix~\ref{sec:proofs}. Building on these lemmas, we now present the proof of Lemma~\ref{lem:adaptive_excitation}.

{\bfseries\itshape Proof of Lemma~\ref{lem:adaptive_excitation}:}
The proof consists of two steps. First, we establish the strong consistency of the parameter estimates \(\hat{\theta}_h\). We then use this result to derive a sharper lower bound on the minimum eigenvalue of the information matrix \(\Lambda_h\).

Step 1: Strong consistency of the parameter estimates. By \eqref{certain} and the definition of the projection set $\mathcal{P}_h$, $\|\hat{\boldsymbol{P}}_h\|
\le \log(h+1)$. Since $\{\gamma_hv_h\}$ is bounded, \eqref{adaptive_incentive} implies \(\|\boldsymbol{P}_h\|=O(\log h)\). Moreover, since $\{\xi_h\}_h$ is bounded, $\{\eta_h\}_h$ is bounded. Furthermore, by \eqref{first-order} and \eqref{eq:stack}, the IONE strategy and the corresponding state trajectory in season $h$ satisfies
\begin{equation}\label{eq: ione_state}
\begin{aligned}
\boldsymbol{U}_h&=-M^{-1}(\mathcal{H}\boldsymbol{P}_h+\eta_h),\\
\boldsymbol{X}_h&=\bar{A}\xi_h+\Gamma \boldsymbol{U}_h.
\end{aligned}
\end{equation} 
It follows that $\|\boldsymbol{U}_h\|+\|\boldsymbol{X}_h\|=O(\log h), a.s.,$ as $h\to \infty.$
Thus, by the definition of $\zeta_{i,h}(t)$ and $\bar \phi_h$ in \eqref{eq:def_zeta} and \eqref{barphi},
\begin{equation}\label{eq:barphi_log}
\max_{\tau\leq h}\sum_{t=1}^{T}\|\zeta_{i,\tau}(t)\|=O(\log h),\; \|\bar\phi_h\|=O(\log h),
\quad\text{a.s.}
\end{equation}
By Lemma~\ref{lem4.2} and \eqref{eq:barphi_log}, there exists a constant $c_{\zeta}>0$ such that
\begin{equation}\nonumber
\lambda_{\min}
\left(I+
\sum_{\tau=1}^{h}
\zeta_{i,\tau}(T-1)\zeta_{i,\tau}(T-1)^{\top}
\right)
\ge
c_{\zeta}
\frac{h^{1-2\gamma}}{(\log h)^{2}},
\; \text{a.s.}
\end{equation}
Substituting this lower bound into Lemma~\ref{lem:phi_excitation}, together with \eqref{eq:barphi_log}, yields a constant \(c_{\phi}>0\) such that
\begin{equation}\label{eq:barphi_all_lower}
\lambda_{\min}
\left(\Lambda_1+
\sum_{\tau=1}^{h}
\bar\phi_{\tau}^{\top}\bar\phi_{\tau}
\right)
\ge
c_\phi
\frac{h^{1-2n\gamma}}{(\log h)^{4n-2}},
\quad\text{a.s.}
\end{equation}
On the other hand, by \eqref{eq:barphi_log},
\begin{equation}\label{eq:barphi_max_rough}
\lambda_{\max}
\left(\Lambda_1+
\sum_{\tau=1}^{h}
\bar\phi_{\tau}^{\top}\bar\phi_{\tau}
\right)
\leq
\sum_{\tau=1}^{h}\|\bar\phi_\tau\|^2
=O\bigl(h(\log h)^2\bigr),
\; \text{a.s.}
\end{equation}
By \eqref{alg}, $\Lambda_{h+1}
=\Lambda_1+
\sum_{\tau=1}^{h}
\bigl(\phi_\tau^\top\phi_\tau-\Xi\bigr).$ Recalling that $\phi_h=\bar\phi_h+\delta_h$, we obtain
from Lemma~\ref{lem:excitation_equivalence} and
\eqref{eq:barphi_max_rough} that, for every
$\eta_0\in(1/2,1-2n\gamma)$,
\begin{equation}\label{difference}
\left\|
\Lambda_{h+1}-\Lambda_{1}-\sum_{\tau=1}^{h}
\bar\phi_\tau^\top\bar\phi_\tau
\right\|
=o\left(h^{\eta_0}
\right),
\quad\text{a.s.}    
\end{equation}
Since $h^{\eta_0}=o\!\left(h^{1-2n\gamma}/(\log h)^{4n-2}\right)$,
it follows from \eqref{eq:barphi_all_lower},
\eqref{eq:barphi_max_rough}, and \eqref{difference} that
\begin{equation}\label{eq:Lambda_rough}
\begin{aligned}
\frac{h^{1-2n\gamma}}{(\log h)^{4n-2}}&=O(\lambda_{\min}(\Lambda_h)),\qquad\text{a.s.},\\
\lambda_{\max}(\Lambda_h)
&=O\bigl(h(\log h)^2\bigr),
\qquad\text{a.s.}
\end{aligned}
\end{equation}
These preliminary bounds give
\begin{equation*}
\begin{aligned}
&\frac{
\bigl(
\lambda_{\max}(\Lambda_h)
\log\log\lambda_{\max}(\Lambda_h)
\bigr)^{1/2}
}{
\lambda_{\min}(\Lambda_h)
}\\
=&
O\left(
h^{-1/2+2n\gamma}
(\log h)^{4n-1}
(\log\log h)^{1/2}
\right)
=o(1),\quad a.s.,  
\end{aligned}
\end{equation*}
where the last relation follows from \(\gamma<1/(4n)\).

Moreover, by Assumption~\ref{assum_noise} and the Borel--Cantelli lemma, for any 
\(\alpha\in(1/\beta,1/2)\), $\|\boldsymbol{\epsilon}_h\|=o(h^\alpha)$ as $h\to \infty.$ Thus,  $\|\delta_h\|=o(h^\alpha)$ as $h\to \infty.$
Since \(\log h=o(h^\alpha)\), \eqref{eq:barphi_log} implies 
\begin{equation}\label{phi_bound}
\|\phi_h\|
\leq
\|\bar\phi_h\|+\|\delta_h\|
=o(h^\alpha),\quad a.s.   
\end{equation}
Thus, all the conditions of Theorem~\ref{thm1} are satisfied, and it follows that $\hat\theta_h\to \theta$ almost surely as $h\to \infty$.

Step 2:  Sharpening the lower bound on $\lambda_{\min}(\Lambda_h)$. By Step 1, $\hat\theta_h\to\theta$ almost surely. Hence, $\{\hat M_h\}_{h\geq1}$ and $\{\hat\eta_h\}_{h\geq1}$ are bounded almost surely. Since $\{\xi_h\}_{h\geq1}$ is bounded, the expression of $\boldsymbol U_h^*$ in Appendix~\ref{sec:social} further implies that $\{\boldsymbol U_h^*\}_{h\geq 1}$ is bounded, and consequently $\{\hat b_h\}_{h\geq1}$ is bounded. The boundedness of $\{\xi_h\}_{h\geq1}$ also ensures that its closure is compact. By Proposition~\ref{pro1}, \(\mathcal H^*_h\) has full row rank on this compact set. Hence, by continuity, $\inf_{h\geq 1}
\sigma_{\min}\bigl(\mathcal H^*_h\bigr)>0,$ which in turn implies that $\mathcal H_h^{*\top}
\bigl(
\mathcal H^*_h\mathcal H_h^{*\top}
\bigr)^{-1}$ is uniformly bounded. This implies that $\{\bar{\boldsymbol{P}}_h\}_{h\geq 1}$ is almost surely bounded. Since the radius $\log(h+1)$ of the projection set $\mathcal{P}_h$ tends to infinity, there exists an almost surely finite $h_1$ such that $\|\bar{\boldsymbol{P}}_h\|
\leq \log(h+1)$ for all
$h\geq h_1.$ Thus,
\begin{equation}\label{eq: projection}
\hat{\boldsymbol{P}}_h=\bar{\boldsymbol{P}}_h,
\qquad
h\geq h_1.
\end{equation}
Thus, by \eqref{adaptive_incentive}, \(\{\boldsymbol{P}_h\}_{h\geq 1}\) is bounded almost surely. Consequently, by \eqref{eq: ione_state},
\begin{equation}\label{eq:trajectory_bounded}
\|\boldsymbol{U}_h\|+\|\boldsymbol{X}_h\|
=O(1),
\;\; \text{a.s.,}\;\; h\to \infty.
\end{equation}
Applying Lemma~\ref{lem:phi_excitation} once again, now using
\eqref{eq:trajectory_bounded}, eliminates the logarithmic loss and yields
\begin{equation}\label{eq:min_barphi}
\lambda_{\min}
\left(\Lambda_1+
\sum_{\tau=1}^{h}
\bar\phi_\tau^\top\bar\phi_\tau
\right)
\geq
\bar c_\phi h^{1-2n\gamma},
\; \text{a.s.},\; h\to \infty,
\end{equation}
for some \(\bar c_\phi>0\). Furthermore,
\begin{equation}\label{eq:max_barphi}
\lambda_{\max}
\left(\Lambda_1+
\sum_{\tau=1}^{h}
\bar\phi_\tau^\top\bar\phi_\tau
\right)
=O(h),
\; \text{a.s.},\; h\to \infty,
\end{equation}
since \(\|\bar\phi_h\|=O(\|\boldsymbol{U}_h\|+\|\boldsymbol{X}_h\|)=O(1),  a.s.\) Similarly, by \eqref{alg}, \eqref{eq:max_barphi}, and Lemma~\ref{lem:excitation_equivalence},
\[
\left\|
\Lambda_{h+1}-\Lambda_{1}-\sum_{\tau=1}^{h}
\bar\phi_\tau^\top\bar\phi_\tau
\right\|
=o\left(
h^{1-2n\gamma}\right),
\; \text{a.s.},\; h\to \infty. 
\]
Therefore, as $h\to \infty$,
\begin{equation}\label{eq:lambda_max_thm2}
h^{1-2n\gamma}=O(\lambda_{\min}(\Lambda_h)),
\quad
\lambda_{\max}(\Lambda_h)
=O(h),
\quad\text{a.s.}    
\end{equation}
This completes the proof of Lemma~\ref{lem:adaptive_excitation}.
\hfill $\square$

{\bfseries\itshape Proof of Theorem \ref{thm2}:}
The condition \eqref{growth_phi} of Theorem~\ref{thm1} has already been verified in \eqref{phi_bound}. Thus, applying the estimation-error bound in Theorem~\ref{thm1} and using
\eqref{eq:lambda_max_thm2} yields
\[
\begin{aligned}
\|\hat\theta_h-\theta\|^2
&=
O\left(
\frac{h\log\log h}
{h^{2-4n\gamma}}
\right)
=O\left(
(\log\log h)h^{-1+4n\gamma}
\right),\;\text{a.s.},
\end{aligned}
\]
which proves Theorem~\ref{thm2}.
\hfill $\square$
\subsection{Proof of Theorem~\ref{thm3}}
For each $h\geq 1$, let
\begin{equation}\nonumber
\boldsymbol{P}_h^{*}=
-\mathcal H_h^{*\top}
\bigl(\mathcal H^*_h\mathcal H_h^{*\top}\bigr)^{-1}
b_h,
\end{equation}
where $\mathcal H^*_h$ and $b_h$ are given in Proposition~\ref{pro1}. By their definitions,
$\boldsymbol{P}^*_h$ satisfies
\begin{equation}\label{eq:pstar_conditions}
\mathcal H\boldsymbol{P}^{*}_h
=
-M\boldsymbol{U}^*_h-\eta_h,\;\;\text{and}\;\;
\mathbf 1_{LT}^{\top}\mathcal G^*_h
\boldsymbol{P}^*_h=0.
\end{equation}
As shown in Step 2 of the proof of Lemma \ref{lem:adaptive_excitation},
 $\mathcal H_h^{*\top}
\bigl(
\mathcal H^*_h\mathcal H_h^{*\top}
\bigr)^{-1}$ is uniformly bounded. Moreover, since $\hat{M}_h$ and $\hat{\eta}_h$ depend locally Lipschitz continuously on the estimate $\hat{\theta}_h$ in a neighborhood of the true parameter $\theta$, we obtain
\begin{equation}\nonumber
\|b_h-\hat{b}_h\|=O(\|\hat\theta_h-\theta\|),\quad a.s.    
\end{equation}
Hence, by 
\eqref{ce} and the above bound of $\|b_h-\hat{b}_h\|$,
\begin{equation}\label{eq:incentive_bound}
\|\bar{\boldsymbol{P}}_h-\boldsymbol{P}^*_h\|
=
O(\|\hat\theta_h-\theta\|),
\qquad \text{a.s.}    
\end{equation}
By \eqref{adaptive_incentive} and \eqref{eq: projection},
$\boldsymbol{P}_h=\bar{\boldsymbol{P}}_h+h^{-\gamma}v_h$ for all $h\geq h_1$. Combining this with \eqref{eq:incentive_bound}, the convergence rate of
$\|\hat{\theta}_h-\theta\|$ in Theorem~\ref{thm2}, and
$\sup_{h\geq 1}\|v_h\|\leq \bar{v}$, 
\begin{equation}\label{eq:p_error_rate}
\begin{aligned}
&\|\boldsymbol{P}_h-\boldsymbol{P}^*_h\|
\leq
\|\bar{\boldsymbol{P}}_h-\boldsymbol{P}^*_h\|
+h^{-\gamma}\|v_h\|\\
=&O(\|\hat{\theta}_h-\theta\|
+h^{-\gamma})=O(\rho_h),
\; \text{a.s.},
\end{aligned}
\end{equation}
where $\rho_h
=
(\log\log h)^{1/2}h^{-1/2+2n\gamma}
+h^{-\gamma}.$ 

By \eqref{eq: ione_state} and \eqref{eq:pstar_conditions}, $\boldsymbol{U}_h-\boldsymbol{U}^*_h
=
-M^{-1}\mathcal H
\bigl(\boldsymbol{P}_h-\boldsymbol{P}^*_h\bigr).$
This yields
\begin{equation}\label{eq:trajectory_error_rate}
\|\boldsymbol{U}_h-\boldsymbol{U}^*_h\|
=O(\|\boldsymbol{P}_h-\boldsymbol{P}^*_h\|)=
O(\rho_h),
\; \text{a.s.}
\end{equation}
Since $\boldsymbol{U}^*_h$ minimizes
$\Phi(\xi_h,\cdot)$, we have $\nabla_{\mathbf u}\Phi
\bigl(\xi_h,\boldsymbol{U}^*_h\bigr)=0.$
Moreover, the Hessian of $\Phi(\xi_h,\cdot)$ is uniformly bounded
over the relevant bounded set. Hence, by a second-order Taylor
expansion and \eqref{eq:trajectory_error_rate},
\[
\Phi(\xi_h,\boldsymbol{U}_h) - \Phi(\xi_h,\boldsymbol{U}^*_h) = O\!\left( \|\boldsymbol{U}_h-\boldsymbol{U}^*_h\|^2 \right) = O(\rho_h^2), \;\; \text{a.s.}.
\]
Summing over $h=1,\ldots,H$ yields
\begin{equation}\nonumber
\mathcal R(H)
=
O\left(
(\log\log H)H^{4n\gamma-1}
+
H^{-2\gamma}
\right),
\quad\text{a.s.}
\end{equation}

Finally, let $\mathcal G_h
=
\operatorname{diag}\!\bigl\{
\mathcal G_{1,h},\ldots,\mathcal G_{L,h}
\bigr\}$ where $\mathcal G_{i,h}
=
\operatorname{diag}\!\bigl\{
g_{i,h}(0),\ldots,g_{i,h}(T-1)
\bigr\},$ $g_{i,h}(t)
=
\begin{bmatrix}
x_h(t)^\top, &
u_{i,h}(t)^\top, &
1
\end{bmatrix}.$ 
By the definition of $\gamma_i$ and the stacking of $\boldsymbol P_h$,
\[
\sum_{i=1}^{L}
\sum_{t=0}^{T-1}
\gamma_i\bigl(x_h(t),u_{i,h}(t), t\bigr)=\mathbf 1_{LT}^\top\mathcal G_h\boldsymbol P_h. 
\]
By
\eqref{eq:pstar_conditions},
$\mathbf 1_{LT}^\top\mathcal G^*_h\boldsymbol{P}^*_h=0$. Hence, 
\begin{equation}\label{eq:budget_proof}
\begin{aligned} &\bigl|\mathbf 1_{LT}^\top\mathcal G_h\boldsymbol P_h\bigr| =\bigl|\mathbf 1_{LT}^\top \bigl(\mathcal G_h\boldsymbol P_h-\mathcal G_h^*\boldsymbol P_h^*\bigr)\bigr|\\ =& O\Bigl( \|\mathcal G_h\|\,\|\boldsymbol P_h-\boldsymbol P_h^*\| + \|\mathcal G_h-\mathcal G_h^*\|\,\|\boldsymbol P_h^*\| \Bigr),\;\; a.s. 
\end{aligned}
\end{equation}
By the definition of $\mathcal G_h$ and $\mathcal G_h^{*}$,
\begin{equation}\label{eq:budget_difference}
\|\mathcal G_h-\mathcal G_h^*\| = O\!\left( \|\boldsymbol U_h-\boldsymbol U_h^*\| \right).
\end{equation}
Moreover, the boundedness of $\{\xi_h\}_{h\geq 1}$ implies that $\{\boldsymbol P_h^*\}_{h\geq 1}$ is bounded. Therefore, by \eqref{eq:p_error_rate}, \eqref{eq:trajectory_error_rate}, \eqref{eq:budget_proof} and \eqref{eq:budget_difference},
\[
\begin{aligned}
|\mathbf 1_{LT}^\top\mathcal G_h\boldsymbol{P}_h|
=O(\rho_h),
\qquad\text{a.s.}
\end{aligned}
\]
Therefore,
\[
\begin{aligned}
|\mathcal B(H)|
&\leq\frac{1}{H}\sum_{h=1}^H|\mathbf 1_{LT}^\top\mathcal G_h\boldsymbol{P}_h| \\
&=
O\left(
(\log\log H)^{1/2}H^{-1/2+2n\gamma}
+H^{-\gamma}
\right),
\qquad\text{a.s.}
\end{aligned}
\]
which proves the desired net-budget bound.
\hfill $\square$
\section{Simulation}\label{sec:sim}
We demonstrate the proposed adaptive incentive algorithm using the Transboundary Pollution~Problem (TPP) \cite{longDynamicGamesTransboundary2012, calvoDynamicModelsEnvironmental2013} between $L$ countries and an international planner. The international planner provides subsidies to the countries, and each country chooses its strategy of pollution mitigation based on its own domestic economic preferences and the international subsidy it receives. The planner's goal is to design the subsidy so as to coordinate the countries' mitigation strategies aligning with a global desirable outcome that reduces overall pollution.

\textbf{The LQ Transboundary Pollution Problem}.
We formulate a discrete dynamic game of the TPP.
The game involves \(L\) countries cumulatively affecting $n$ pollution sinks, such as seas or waterways that border them. Each country selects a pollution abatement action $u_i \in \mathbb{R}^m$, which affects the global pollution profile $x\in \mathbb{R}^n$ through the dynamics
\begin{equation}\label{eq:dynamics_sim}
x(t+1)
=
A x(t)+E-\sum_{i=1}^{L}B_i u_i(t),
\quad x(0)=\xi,
\end{equation} 
where \(x(t)\in\mathbb R^n\) represents the system state, 
\(u_i(t)\in\mathbb R^m\) is the mitigation strategy selected by country \(i\), and
\(E\in\mathbb R^n\) is a vector of baseline pollution inflow. The matrix $A$ represents the pollution absorption and mixing between sinks, and \(B_iu_i\) is the marginal influence due to the abatement of country \(i\)\footnote{Model \eqref{eq:dynamics_sim} is a discrete-time LQ formulation of stock-pollution dynamics in transboundary pollution
games \cite{pollution_li}.} 

Assume the international planner provides each country $i$ with an affine subsidy at time t, as specified in \eqref{incentive}. Given the incentive $\{p_i(t)\}_t$, country $i$'s cost is given by \eqref{eq:total_cost}, where $Q_i \succeq 0$ represents the weight on deviations from the desired state, while $R_i \succ 0$ represents the weight on its abatement effort.

The \emph{social welfare loss} of the international planner $\Phi(\xi, \boldsymbol{U})$ is specified in \eqref{social_cost},
where $Q^{*}\succeq0$ represents the planner's weight on the system state,
while $R_i^{*}\succ0$ represents the global cost associated with country
$i$'s abatement effort.
The international planner does not know the countries' preference parameters
\(\{Q_i,R_i\}_{i=1}^L\), as those would be affected by country-specific environmental priorities, abatement costs, and environmental policy, which are typically private or difficult to quantify. As a result, the planner calls for the repeated interaction with the countries over seasons (months) to estimate their private preference, then issues an adaptive subsidy $\{p_i(t)\}_t$ to each country to minimize the social welfare cost. The planner observes the agents' state-action trajectory and uses this information to update the incentive design for the next season. We validate the performance of the proposed adaptive algorithm through the following numerical experiments.

\textbf{Simulation results.} We choose $n=2, L=3, m=1$, and set
\[
A=\begin{bmatrix}
0.75 & 0.18\\
0.14 & 0.8
\end{bmatrix},\quad E=
\begin{bmatrix}
1\\
1
\end{bmatrix},
\]
and
\[
B_1=
\begin{bmatrix}
1\\
0.5
\end{bmatrix},\quad B_2=
\begin{bmatrix}
0.4\\
1
\end{bmatrix},
\quad
B_3=
\begin{bmatrix}
0.5\\
0.3
\end{bmatrix}.
\]
Each season has horizon length \(T=15\), and the initial pollution stock is
fixed as $\xi=[2, 2]^{\top}$ for each $h\geq 1.$
The true preference matrices are $Q_1=\operatorname{diag}(1, 0.5)$, $Q_2=\operatorname{diag}(0.4, 1)$, $Q_3=\operatorname{diag}(0.1, 0.1)$, 
$R_1=1, R_2=3,$ and $ R_3=0.5.$ The social cost parameters are chosen as $Q^{*}=\operatorname{diag}(1, 1)$ and $R_1^{*}=1, R_2^{*}=3, R_3^{*}=0.5.$ At each season \(h\), the planner observes noisy
state and action trajectories where noise $\epsilon_h$ is Gaussian with variance $10^{-5}$. The adaptive incentive is given by \eqref{adaptive_incentive}, where \(v_h\) is a Rademacher probing signal,  $\gamma_h=h^{-0.03}$.

We compare the proposed Algorithm~1 against a certainty-equivalence baseline, denoted IOC-CE, which employs the inverse optimal control or game estimation~\cite{ inverse_game_Observations} to repeatedly estimate the countries' preference  from all trajectories collected so far, and then issues the incentive corresponding to the current estimate as given in \eqref{adaptive_incentive}, without additional exploratory probing. 
Figure~\ref{fig:estimation_error_comparison} shows that the estimation error \(\|\hat\theta_h-\theta\|^2\) under Algorithm~1 steadily decreases with the season $h$, whereas IOC-CE stagnates at a much larger level, since its incentive is not designed to actively excite the trajectories for identification. This is confirmed by the minimum eigenvalue of the information matrix $\lambda_{min}(\Lambda_h-\Lambda_1)$, which grows steadily under Algorithm~1 but remains essentially flat under IOC-CE, highlighting the importance of active excitation and a principled exploration–exploitation trade-off for closed-loop incentive design.

\begin{figure}[htbp]
    \centering    \includegraphics[width=0.5\textwidth]{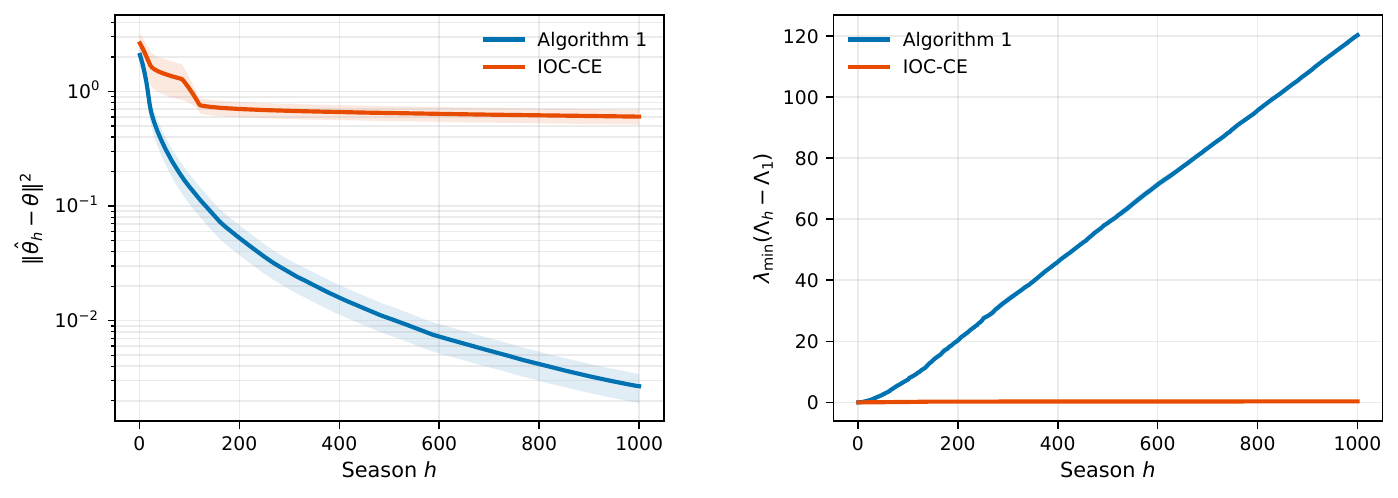}
\caption{Comparison of learning performance and excitation. Left: squared parameter estimation error. Right: minimum eigenvalue of the accumulated information matrix. Solid lines represent the mean over independent runs (each with independently sampled observation noise and exploration signal), and shaded regions indicate one standard error of the mean.}
\label{fig:estimation_error_comparison}
\end{figure}

Figure~\ref{fig:subsidy} plots the trajectories of the average social-cost regret $\mathcal{R}(h)$ and the average net budget $\mathcal{B}(h)$. Both quantities approach zero, consistent with the theoretical guarantees established in Theorem~\ref{thm3}. 

\begin{figure}[htbp]
    \centering    \includegraphics[width=0.5\textwidth]{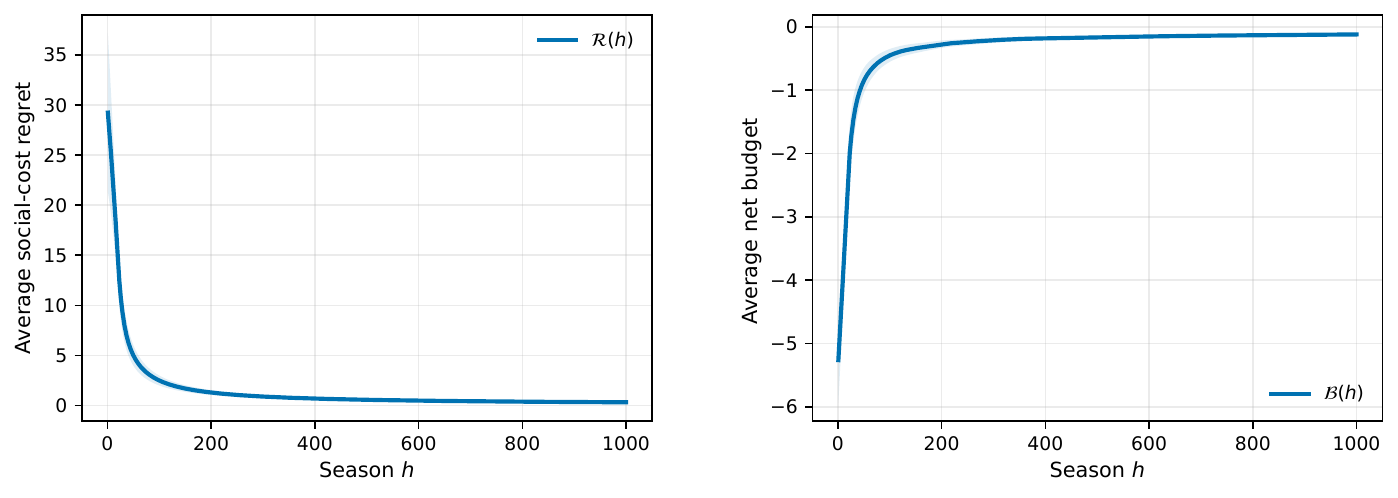}
\caption{Trajectories of the average social-cost regret $\mathcal{R}(h)$ (left) and the average net budget $\mathcal{B}(h)$ (right) under the proposed Algorithm 1. The solid curves show the empirical means, and the shaded regions indicate one standard deviation.}
 \label{fig:subsidy}
\end{figure}

        
	\section{Concluding remarks}\label{se6}	
Under the information asymmetry, we studied adaptive incentive design between a planner and $L$ agents that interact through a dynamic game, in which the agents play an open-loop Nash equilibrium strategy. The agents' costs depend on preference parameters over the system state and actions, which are unknown to the planner.
We proposed an online learning algorithm to estimate agents’ unknown preferences
 and established its strong consistency under a non-PE condition on the observed trajectories.
 Building on this, we developed an adaptive incentive algorithm by using certainty-equivalent design and decaying random probing, and proved that it achieves sublinear regret and net budget relative to the oracle policy. Future work includes extending the framework to closed-loop Nash equilibrium strategies, partial observations, and nonlinear dynamics with nonlinear incentive mechanisms.

\bibliographystyle{IEEEtran}
\bibliography{reference}

\appendices
\section{}\label{sec:social}
Below, we provide an explicit characterization of the socially optimal solution \(\boldsymbol{U}_h^*\) for each season \(h\).

Define
$B=[B_1,\ldots,B_L]$ and $R^*=\operatorname{diag}\{R_1^*,\ldots,R_L^*\}$. Let
$V^{*}(T)=0, r^{*}(T)=0$. For $t=T-1,\ldots,0$, define
\begin{equation*}
\begin{aligned}
K^*(t)
&=
\bigl(R^*+B^\top V^*(t+1)B\bigr)^{-1}
B^\top V^*(t+1)A,\\
V^*(t)
&=
Q^*+A^\top V^*(t+1)A
-K^{*\top}(t)B^\top V^*(t+1)A,\\
r^*(t)
&=
Q^*x^{d*}+A^\top r^*(t+1)
-K^{*\top}(t)B^\top r^*(t+1).
\end{aligned}
\label{eq:riccati_recursion}
\end{equation*}
For a given initial state $\xi_h$, write the socially optimal
strategy defined in \eqref{social_solution} as
$\boldsymbol U_h^*
=[\boldsymbol u_{1,h}^{*\top},\ldots,
\boldsymbol u_{L,h}^{*\top}]^{\top}$, where
$
\boldsymbol u_{i,h}^*
=[u_{i,h}^*(0)^{\top},\ldots,u_{i,h}^*(T-1)^{\top}]^{\top}
$.
For $t=0,\ldots,T-1$, the joint control
$\boldsymbol u_h^*(t)
=[u_{1,h}^*(t)^{\top},\ldots,u_{L,h}^*(t)^{\top}]^{\top}$
is given by
\begin{equation*}
\begin{aligned}
\boldsymbol{u}_h^{*}(t)&={}\bigl(R^*+B^\top V^*(t+1)B\bigr)^{-1}
B^\top r^*(t+1)-K^{*}(t)x^*_h(t).
\end{aligned}
\label{eq:optimal_control}
\end{equation*}
The corresponding optimal state trajectory is generated recursively by $x^*_h(t+1)
=
Ax^*_h(t)+B\boldsymbol{u}_h^{*}(t)$ with $x^*_h(0)=\xi_h$.

\section{}\label{sec:proofs}

\noindent{\bfseries\itshape Proof of Lemma~\ref{lem_ne}:}
Define
\begin{equation}\label{uniquness}
\begingroup
\setlength{\arraycolsep}{2pt}
\small
M=
\begin{bmatrix}
\Gamma_1^\top \bar Q_1 \Gamma_1+\bar R_1
&
\cdots
&
\Gamma_1^\top \bar Q_1 \Gamma_L
\\
\vdots
&
\ddots
&
\vdots
\\
\Gamma_L^\top \bar Q_L \Gamma_1
&
\cdots
&
\Gamma_L^\top \bar Q_L \Gamma_L+\bar R_L
\end{bmatrix},
\endgroup
\end{equation}
where
\begin{equation}\label{eq:gamma}
\Gamma_i=
\begin{bmatrix}
0_{n\times m(T-1)} & 0_{n\times m}\\
\mathcal L_{T-1}(A,B_i) & 0_{n(T-1)\times m}
\end{bmatrix},    
\end{equation}
and $\bar Q_i=I_T\otimes Q_i, 
\bar R_i=I_T\otimes R_i.$
The operator $\mathcal L_{\cdot}(\cdot,\cdot)$ is defined
in~\eqref{eq: operator_L}. Moreover, define $\eta=[\eta_1^{\top}, \cdots \eta_L^{\top}]^{\top}$, $\eta_i =
\Gamma_i^\top
\Bigl[
(I_T \otimes Q_i)\bar A\,\xi
-\mathbf{1}_T \otimes (Q_i x_i^d)
\Bigr],$ and
\begin{equation}\label{eta}
\begin{gathered}
\bar A =
\begin{bmatrix}
I_n & A^\top & \cdots & (A^{T-1})^\top
\end{bmatrix}^{\top},
\\[4pt]
S_x =
I_T \otimes
\begin{bmatrix}
I_n & 0_{n\times m} & 0_{n\times 1}
\end{bmatrix},
\\[4pt]
S_u =
I_T \otimes
\begin{bmatrix}
0_{m\times n} & I_m & 0_{m\times 1}
\end{bmatrix},
\\[4pt]
\mathcal H
=
\frac{1}{2}
\operatorname{diag}\!
\left\{
\Gamma_1^\top S_x+S_u,\ldots,
\Gamma_L^\top S_x+S_u
\right\}.
\end{gathered}
\end{equation}
For a strategy profile
$\boldsymbol{U}=[\boldsymbol{u}_1^\top,\ldots,\boldsymbol{u}_L^\top]^\top$, the corresponding stacked state trajectory $\boldsymbol{X}=[x(0)^\top,\ldots,x(T-1)^\top]^\top$ satisfies
\begin{equation}\label{eq:stack}
X
=
\bar A\xi+\Gamma\boldsymbol{U},
\end{equation}
where $\Gamma=[\Gamma_1, \ldots, \Gamma_L]$.
Let
$\bar Q_i=I_T\otimes Q_i$ and
$\bar R_i=I_T\otimes R_i$.
Substituting \eqref{eq:stack} into player $i$'s cost and
differentiating with respect to $\mathbf u_i$ gives
\begin{equation}\label{eq:foc_i}
\begin{aligned}
0
={}&
2\Gamma_i^\top\bar Q_i
\left(
\bar A\xi+\Gamma\boldsymbol{U}
-\mathbf 1_T\otimes x_i^d
\right)
+2\bar R_i\mathbf u_i \\
&+
\left(\Gamma_i^\top S_x+S_u\right)\mathbf p_i .
\end{aligned}
\end{equation}
Moreover, the Hessian of player \(i\)'s cost with respect to \(\boldsymbol{u}_i\) is $2(\Gamma_i^\top\bar Q_i\Gamma_i+\bar R_i)\succ0.$ Hence, for any fixed strategies of the other players,
player \(i\)'s cost is strictly convex in \(\boldsymbol{u}_i\).
Therefore, \eqref{eq:foc_i} is both necessary and sufficient for
\(\boldsymbol{u}_i\) to be player \(i\)'s unique best response. Stacking \eqref{eq:foc_i} over $i=1,\ldots,L$ yields $M\boldsymbol{U}
=
-\eta-\mathcal H\boldsymbol{P}.$
By Assumption~\ref{assum_ne}, $M$ is nonsingular. Thus, \eqref{first-order} follows. \hfill $\square$

\noindent{\bfseries\itshape Proof of Proposition \ref{pro1}:}
By \eqref{first-order} and the nonsingularity of $M$, an incentive
$\boldsymbol{P}_h$ satisfies \eqref{eq:equilibrium_social_cost} if and only if
\begin{equation}\label{eq:optimal_strategy}
 M\boldsymbol{U}^*_h=-\mathcal H\boldsymbol{P}_h
-\eta_h.   
\end{equation}
If the incentive \(\boldsymbol{P}_h\) satisfies \eqref{eq:optimal_strategy}, then the total net transfer can be written as
\[
\sum_{i=1}^{L}\sum_{t=0}^{T-1}
\gamma_i\bigl(x^*_h(t),u_{i,h}^*(t), t\bigr)
=
\mathbf 1_{LT}^{\top}\mathcal G^*_h\boldsymbol{P}_h.
\]
It follows that an incentive satisfies \eqref{eq:equilibrium_social_cost} and \eqref{eq:min_payment_incentive} if and only if it
satisfies both \eqref{eq:optimal_strategy} and $1_{LT}^{\top}\mathcal G^*_h\boldsymbol{P}_h=0.$ Equivalently, 
\[
\mathcal H^*_h\boldsymbol{P}_h=-b_h,
\]
where $\mathcal H_h^{*}$ and $b_h$ are given in \eqref{eq:optimal_p}.
By the definitions of $S_x$ and $S_u$, each diagonal block
$\Gamma_i^\top S_x+S_u$ of $\mathcal H$ contains an $I_{mT}$ submatrix, so \(\mathcal H\) has full row rank. Moreover, the columns associated with the constant incentive
coefficients are zero in \(\mathcal H\) but equal to one in
\(\mathbf 1_{LT}^{\top}\mathcal G_h^*\). Thus, the last row of \(\mathcal H_h^*\) is linearly independent of
the rows of \(\mathcal H\), and hence \(\mathcal H_h^*\) has full row rank. Therefore, the complete set of
solutions is
\[
\boldsymbol{P}_h
\in
-\mathcal H_h^{*\top}
\bigl(
\mathcal H^*_h\mathcal H_h^{*\top}
\bigr)^{-1}
b_h
+
\ker\bigl(\mathcal H_h^{*}\bigr),
\]
which proves Proposition \ref{pro1}.
\hfill $\square$

\noindent{\bfseries\itshape Proof of Lemma~\ref{lem31}:}
For each agent $i$, partition $p_{i,h}(t)
=
[
p_{i,h}^x(t)^{\top},p_{i,h}^u(t)^{\top},p_{i,h}^0(t)]^{\top}
,$ where $p_{i,h}^x(t)\in\mathbb R^n$,
$p_{i,h}^u(t)\in\mathbb R^m$ and $p_{i,h}^0(t)\in\mathbb R$. By the Pontryagin maximum principle,
there exists a costate sequence
$\{\lambda_{i,h}(t)\}_{t=0}^{T-1}$ satisfying $\lambda_{i,h}(T-1)=0$ and
\begin{align}
\label{eq:pmp_terminal}
\lambda_{i,h}(t-1)
&=
Q_i\bigl(x_h(t)-x_i^d\bigr)
+A^\top\lambda_{i,h}(t)
+\frac12 p_{i,h}^x(t), \nonumber\\
-R_i u_{i,h}(t)
&=B_i^\top\lambda_{i,h}(t)
+\frac{1}{2} p_{i,h}^u(t),
\qquad t=0,\ldots,T-1. \nonumber
\end{align}
Stacking these relations together with the terminal condition
and defining
$\lambda_{i,h}
=[\lambda_{i,h}(0)^\top,\ldots,\lambda_{i,h}(T-1)^\top]^\top$,
we obtain
\begin{equation}\label{eq:adjoint}
\begin{gathered}
\mathcal A\lambda_{i,h}
=
(\mathcal{J}_T\otimes Q_i)
\left(
\boldsymbol{X}_h-\mathbf 1_T\otimes x_i^d
\right)
+
\frac{1}{2}(\mathcal{J}_T\otimes I_n)S_x\mathbf{p}_{i,h},  \\
-(I_T\otimes R_i)\boldsymbol{u}_{i,h}
=
(I_T\otimes B_i^\top)\lambda_{i,h}
+
\frac{1}{2}S_u\mathbf{p}_{i,h}.
\end{gathered}
\end{equation}
where $\mathcal A=I_{nT}-(\mathcal{J}_T\otimes A^\top),$ and
\[
\mathcal{J}_T =
\begin{bmatrix}
0 & 1 & \cdots & 0\\
\vdots & \vdots & \ddots & \vdots\\
0 & 0 &\cdots & 1\\
0 & 0& \cdots& 0
\end{bmatrix}
\in \mathbb{R}^{T\times T}.
\]
Under the adopted stacking convention,
$\Gamma_i^\top=
(I_T\otimes B_i^\top)\mathcal A^{-1}(\mathcal{J}_T\otimes I_n).$ Let $\bar Q_i=I_T\otimes Q_i$,
$\bar R_i=I_T\otimes R_i$, and 
$\bar{x}_i^d=Q_i x_i^d$. 
Eliminating $\lambda_{i,h}$ from \eqref{eq:adjoint} yields
\begin{equation}\label{eq:regression}
\bar R_i\boldsymbol{u}_{i,h}
+\Gamma_i^{\top}\bar Q_i\boldsymbol{X}_h
-\Gamma_i^{\top}(\mathbf{1}_T\otimes \bar{x}_i^d)=
-\frac{1}{2}\left(\Gamma_i^{\top}S_x+S_u\right)\boldsymbol{p}_{i,h}.
\end{equation}

It remains to rewrite the three terms on the left-hand side of~\eqref{eq:regression} as linear functions of the unknown parameters $R_i, Q_i$ and $\bar{x}_i^{d}$. By the standard property of the Kronecker product and vectorization, we have
\begin{equation}\label{kronecker}
    \begin{aligned}
&\bar R_i\boldsymbol{u}_{i,h}=\mathcal{K}_{m,T}(\boldsymbol{u}_{i,h})\operatorname{vec}(R_i),\\&\bar Q_i\boldsymbol{X}_{h}=\mathcal{K}_{n,T}(\boldsymbol{X}_{h})\operatorname{vec}(Q_i),\\
&(\mathbf{1}_{T}\otimes \bar{x}_i^d)=(\mathbf{1}_{T}\otimes I_n)\bar{x}_i^d.       
    \end{aligned}
\end{equation}
For each $i\in[L]$, define
\begin{equation}\label{eq:bar_phi}
\bar{\phi}_{i,h}=\begin{bmatrix}    \mathcal{K}_{m,T}(\boldsymbol{u}_{i,h})\mathscr{D}_m & \Gamma_i^{\top} \mathcal{K}_{n,T}(\boldsymbol{X}_{h})\mathscr{D}_n & -\Gamma_i^{\top}(\mathbf{1}_{T}\otimes I_n)
\end{bmatrix}.    
\end{equation}
The matrix $\bar{\phi}_h$ defined in
\eqref{barphi} can then be written as $\bar{\phi}_h=\operatorname{diag}\{\bar{\phi}_{1,h},\cdots \bar{\phi}_{L,h}\}$.
Combining~\eqref{eq:regression} with \eqref{kronecker}, and using $\operatorname{vec}(R_i)=\mathscr{D}_m\operatorname{vech}(R_i), \operatorname{vec}(Q_i)=\mathscr{D}_n\operatorname{vech}(Q_i)$, where \(\mathscr D_m\) and \(\mathscr D_n\) are the duplication matrices of
appropriate dimensions, we obtain
$y_{i,h}=\bar{\phi}_{i,h}\theta_i.$ Since $E_x\boldsymbol{Z}_{h}=\boldsymbol{X}_{h}+E_x\boldsymbol{\epsilon}_h$ and $E_{u_i}\boldsymbol{Z}_{h}=\boldsymbol{u}_{i,h}+E_{u_i}\boldsymbol{\epsilon}_h$, we have $\phi_{i,h}=\bar{\phi}_{i,h}+\delta_{i,h}.$
Thus, 
\[
(\phi_{i,h}-\delta_{i,h})\theta_i=y_{i,h}.
\]
Stacking these identities over all agents $i=1,\ldots,L$ yields $(\phi_h-\delta_h)\theta=y_h$
which proves Lemma~\ref{lem31}.
\hfill $\square$

{\bfseries\itshape Proof of Lemma~\ref{lem:preliminary_bounds}:}
We first establish lower bounds for $g_h$ and $\lambda_{\max}(\Lambda_h)$.
Fix any $i\in[L]$. By the block structure of $\Gamma_i$ in
\eqref{eq:gamma},
$\Gamma_i^\top(\mathbf 1_T\otimes I_n)$ contains $B_i^\top$
as one of its blocks. Since $(A,B_i)$ is controllable,
$B_i\neq0$, and hence
$\Gamma_i^\top(\mathbf 1_T\otimes I_n)\neq0$.
Thus, there exists a unit vector $q_i\in\mathbb R^n$ such that $\Gamma_i^\top(\mathbf 1_T\otimes I_n)q_i\neq0.$
Define
\[
a_i
=
\operatorname{col}\!\left\{
0_{(i-1)d_\theta},
0_{\frac{m(m+1)}{2}},
0_{\frac{n(n+1)}{2}},
q_i,
0_{(L-i)d_\theta}
\right\}.
\]
By the definitions of $\phi_h$ and $\delta_h$ in~\eqref{phi},
\[
\delta_h a_i=0,
\qquad
\|\phi_h a_i\|^2
=
\left\|
\Gamma_i^\top(\mathbf 1_T\otimes I_n)q_i
\right\|^2>0.
\]
Thus, by $\bar{\phi}_h=\phi_h-\delta_h$,
\[
\|\bar\phi_h\|^2
\ge
\|\bar\phi_h a_i\|^2
=
\|\phi_h a_i\|^2>0,
\]
and hence $g_h
=
\sum_{j=1}^h\|\bar\phi_j\|^2
\ge h\left\|
\Gamma_i^\top(\mathbf 1_T\otimes I_n)q_i
\right\|^2.$
Therefore, $h=O(g_h), a.s.$ Moreover, since $\Xi\succeq0$ and
$a_i^\top\Xi a_i
=
\mathbb E[\|\delta_h a_i\|^2\mid\mathcal F_{h-1}]=0$,
we also have $\Xi a_i=0$. By the recursion of $\Lambda_h$ in~\eqref{alg},
\begin{equation}\label{eq:lambda_recursion}
\Lambda_{h+1}=
\Lambda_1+\sum_{j=1}^h\phi_j^\top\phi_j
-h\Xi.    
\end{equation}
 Hence, $\lambda_{\max}(\Lambda_{h+1})
\ge
a_i^\top\Lambda_1a_i+h\left\|
\Gamma_i^\top(\mathbf 1_T\otimes I_n)q_i
\right\|^2,$ which proves~$h=O(\lambda_{\max}(\Lambda_{h+1}))$. 

Furthermore, by
$\bar\phi_h=\phi_h-\delta_h$, we have $g_h
\leq
2\sum_{j=1}^{h}\|\phi_j\|^2
+
2\sum_{j=1}^{h}\|\delta_j\|^2.$
Taking traces of \eqref{eq:lambda_recursion} gives
\[
\sum_{j=1}^h\|\phi_j\|^2
=
O\!\left(
\lambda_{\max}(\Lambda_{h+1})+h
\right),
\quad\text{a.s.}
\]
Since $\sum_{j=1}^{h}\|\delta_j\|^2=O(h)$ almost surely, 
\[
g_h
=
O\!\left(
\lambda_{\max}(\Lambda_{h+1})+h
\right)
=
O\!\left(
\lambda_{\max}(\Lambda_{h+1})
\right),
\quad\text{a.s.},
\]

Next, we establish the growth rates of
$\bar\phi_h$. Choose $\bar\alpha\in
\left(
\max\left\{\alpha,\frac{1}{\beta}\right\},
\frac12
\right).$ By Assumption~\ref{assum_noise} and the conditional
Borel--Cantelli lemma, $\|\epsilon_h\|
=
o(h^{\bar\alpha})$ almost surely as $h\to \infty.$
Since $\delta_h$ depends linearly on $\epsilon_h$ through fixed
matrices, there exists a constant $C>0$ such that
$\|\delta_h\|\leq C\|\epsilon_h\|$. Hence,
\begin{equation}\label{eq:delta_growth}
\|\delta_h\|
=
o(h^{\bar\alpha}),
\; \text{a.s.,}\quad h\to \infty.
\end{equation}
Moreover, \eqref{growth_phi} gives
$\|\phi_h\|=o(h^\alpha)=o(h^{\bar\alpha})$ because
$\bar\alpha>\alpha$. Therefore, using
$\bar\phi_h=\phi_h-\delta_h$,
\begin{equation}\label{eq:barphi_growth}
\|\bar\phi_h\|
\leq
\|\phi_h\|+\|\delta_h\|
=
o(h^{\bar\alpha}),
\; \text{a.s.,}\quad h\to \infty.
\end{equation}  
\hfill $\square$

\noindent{\bfseries\itshape Proof of Lemma~\ref{lem4}:}
By suitable row and column permutations, there exist permutation
matrices $P_1$ and $P_2$ such that
\[
P_1^\top\mathscr W_hP_2
=
\bar{\mathscr W}_h
:=
\begin{bmatrix}
W_h & \tilde{W}_h^{(1)} & \cdots & \tilde{W}_h^{(n-1)}\\
0   & W_h       & \cdots & \tilde{W}_h^{(n-2)}\\
\vdots & \vdots & \ddots & \vdots\\
0 & 0 & \cdots & W_h
\end{bmatrix},
\]
where, for $k=1,\ldots, n-1$,
$\tilde{W}_h^{(k)}=[\tilde{w}_{ij}^{(k)}I_m]_{i,j}$
is the block matrix with $(i,j)$th block
$\tilde{w}_{ij}^{(k)}I_m$.
By \eqref{eq:d_h}, $\lambda_{\min}(W_h^\top W_h)\ge\rho_h,$
so
$W_h^\dagger=(W_h^\top W_h)^{-1}W_h^\top$ satisfies
$\|W_h^\dagger\|
\le \rho_h^{-1/2}.$ Moreover, the bounds $|\tilde{w}_{ij}^{(k)}|\le C_h$ and $|w_{ij}|\le C_h$ imply
$\|\tilde{W}_h^{(k)}\|\le C_h\sqrt{hr}$ and $\|W_h\|\le C_h\sqrt{hr}$. Hence
\begin{equation}\label{eq:mathcal_n}
(I_n\otimes W_h^\dagger)\bar{\mathscr W}_h
=
I+\mathcal N_h,
\quad
\|\mathcal N_h\|
\le
C\,C_h\sqrt{\frac{h}{\rho_h}},
\end{equation}
where $\mathcal N_h$ is strictly block upper triangular and
$\mathcal N_h^n=0$.
Since $(I+\mathcal N_h)^{-1}
=
\sum_{q=0}^{n-1}(-\mathcal N_h)^q,$ and
$\rho_h\le\lambda_{\max}(W_h^\top W_h)
\le h r C_h^2$,
there exists $C_0>0$, independent of $h$, such that
\begin{equation}\label{eq: bound_n}
\|(I+\mathcal N_h)^{-1}\|
\le
C_0
\left(
\frac{hC_h^2}{\rho_h}
\right)^{(n-1)/2}.
\end{equation}
Consequently, by \eqref{eq:mathcal_n} and \eqref{eq: bound_n}, for every $x$, $\|\bar{\mathscr W}_h x\|
\ge
\sqrt{\rho_h}\,
\|(I+\mathcal N_h)x\|\ge
C_0^{-1}
h^{-(n-1)/2}
\rho_h^{n/2}
C_h^{-(n-1)}
\|x\|.$
Since permutation matrices preserve singular values,
\[
\lambda_{\min}
\bigl(\mathscr W_h^\top\mathscr W_h\bigr)
\ge
\mu
\frac{h^{1-n}\rho_h^n}{C_h^{2n-2}}
\]
for some constant $\mu>0$ independent of $h$.
\hfill $\square$

{\bfseries\itshape Proof of Lemma \ref{lem:phi_excitation}:}
For \(j=1,\ldots,n\), let \(\mathscr S_j\in\mathbb R^{mn\times mn}\) be the block matrix whose
\((p,q)\)-th \(m\times m\) block is
\[
[\mathscr S_j]_{pq}
=\begin{cases}
I_m,& q-p=n-j,\\
0_m,&\text{otherwise}.
\end{cases}
\]
In particular, $\mathscr S_n=I_{mn}$. For $i\in [L]$ and $h\geq 1$, define
{\small \[
\begin{aligned}
\mathcal X_{i,h}
={}
\Bigl[
&\sum_{j=1}^{n}
u_{i,h}(T-n+j-2)^\top
\otimes\mathscr S_j,
\\
&\quad 
\sum_{j=1}^{n}
x_h(T-n+j-1)^\top
\otimes\mathscr S_j,\;\;
-\sum_{j=1}^{n}\mathscr S_j
\Bigr]^{\top}.
\end{aligned}
\]}
and $\mathscr X_{i,h}=[\mathcal X_{i,1}, \ldots \mathcal X_{i,h}]^{\top}$. Moreover,
for each $i\in [L],$ let $\widetilde{\zeta}_{i,h}=[\zeta_{i,1}(T-1), \ldots \zeta_{i,h}(T-1)]^{\top}\otimes I_m$.
By \eqref{zeta}, there exists a sufficiently large $h_0$ such that, for all $h\ge h_0$,
{\small
\begin{equation}\label{eq: tilde_zeta}
\lambda_{\min}
\left(
\widetilde{\zeta}_{i,h}^{\top}\widetilde{\zeta}_{i,h}
\right)
=\lambda_{\min}
\left(
\sum_{\tau=1}^{h}
\zeta_{i,\tau}(T-1)\zeta_{i,\tau}(T-1)^{\top}
\right)
\ge \frac{1}{2}\rho_h.   
\end{equation}}
Applying Lemma~\ref{lem4} with $W_h=\tilde{\zeta}_{i,h}$
and $\mathscr{W}_h=\mathscr X_{i,h}$, and using
\eqref{eq: tilde_zeta} and \eqref{trajectory_bound},
we conclude that there exists a constant $\bar\mu_i>0$ such that
{\small
\begin{equation}\label{mathcalx}
\lambda_{\min}
\left(
\mathscr X_{i,h}^{\top}\mathscr X_{i,h}
\right)=\lambda_{\min}
\left(
\sum_{\tau=1}^{h}
\mathcal X_{i,\tau}\mathcal X_{i,\tau}^{\top}
\right)
\ge
\bar \mu_i\frac{h^{1-n}\rho_h^{n}}{C_h^{2n-2}},
\quad\text{a.s.}    
\end{equation}
}

   Partition \(\bar\phi_{i,h}\) in \eqref{eq:bar_phi} into \(T\) block rows: 
\[
\bar\phi_{i,h}
=
[\bar\phi_{i,h}^{(1)\top},\ldots,\bar\phi_{i,h}^{(T)\top}]^{\top},
\]
where 
$
\bar\phi_{i,h}^{(j)}
\in
\mathbb R^{m\times [m(m+1)/2+n(n+1)/2+n]},
j=1,\ldots T.
$ By the block upper-triangular structure of $\Gamma_i^\top$ and the
definition of $\mathcal K_{\cdot,T}(\cdot)$, the $j$-th block row of
$\bar\phi_{i,h}$ is given by
{\small
\begin{equation}\label{eq:barphi_block}
\begin{aligned}
\bar\phi_{i,h}^{(j)}
={}&
\Bigl[
(u_{i,h}(j-1)^\top\otimes I_m)\mathscr D_m,
\\
&\quad
(
\sum_{\tau=j}^{T-1}
x_h(\tau)^\top
\otimes B_i^\top(A^{\tau-j})^\top)\mathscr D_n,
\\
&\quad
-\sum_{\tau=j}^{T-1}
B_i^\top(A^{\tau-j})^\top
\Bigr],
\end{aligned}
\end{equation}}
for $j=1,\ldots,T-1$. We select the \(n\) block rows of \(\bar\phi_{i,h}\), indexed by $k=T-n, \ldots T-1$, and define
\[
\tilde\phi_{i,h}
=\left[\bar\phi_{i,h}^{(T-n)\top},\ldots,
\bar\phi_{i,h}^{(T-1)\top}\right]^{\top}.
\]
Define $\mathscr E=\begin{bmatrix}
0, 0,
\ldots
I_m
\end{bmatrix}^{\top}
\in\mathbb R^{mn\times m}$, 
\[
\mathscr{C}_i
=\begin{bmatrix}
A^{n-1}B_i,\
A^{n-2}B_i,
\ldots\
B_i
\end{bmatrix}^{\top}
\in\mathbb R^{mn\times n},
\]
and
\[
\mathcal C_i
=\begin{bmatrix}
(I_m\otimes\mathscr E)\mathscr D_m &0&0\\
0&(I_n\otimes\mathscr C_i)\mathscr D_n&0\\
0&0&\mathscr C_i
\end{bmatrix}.
\]
By the definition of $\mathscr S_j, \mathscr C_i$, and $\mathscr E$,
\[
[\mathscr S_j\mathscr E]_p
=\begin{cases}
I_m,&p=j,\\
0,&p\neq j,
\end{cases}
\]
and
\[
[\mathscr S_j\mathscr C_i]_p
=\begin{cases}
B_i^\top(A^{j-p})^\top,&p\le j,\\
0,&p>j,
\end{cases}
\]
where \([\cdot]_p\) denotes the \(p\)-th \(m\)-row block. Consequently,
stacking the preceding \(n\) block rows gives
{\small
\begin{equation}\label{eq:tildephi_expansion}
\begin{aligned}
\tilde\phi_{i,h}
={}&
\Biggl[
\left(
\sum_{j=1}^{n}
u_{i,h}(T-n+j-2)^\top
\otimes(\mathscr S_j\mathscr E)
\right)\mathscr D_m,
\\
&\quad
\left(
\sum_{j=1}^{n}
x_h(T-n+j-1)^\top
\otimes(\mathscr S_j\mathscr C_i)
\right)\mathscr D_n,
\\
&\quad
-\sum_{j=1}^{n}\mathscr S_j\mathscr C_i
\Biggr].
\end{aligned}
\end{equation}}
Using $(a^\top\otimes \mathscr S_j\mathscr E)
=(a^\top\otimes\mathscr S_j)
(I_m\otimes\mathscr E)
$ and, similarly, $(b^\top\otimes\mathscr S_j\mathscr C_i)
=(b^\top\otimes\mathscr S_j)
(I_n\otimes\mathscr C_i),$ then~\eqref{eq:tildephi_expansion} admits the factorization
\begin{equation}\label{eq:tildephi_factor}
\tilde\phi_{i,h}
=\mathcal X_{i,h}^{\top}\mathcal C_i.
\end{equation}
By controllability of \((A,B_i)\), \(\mathscr C_i\) has full column rank.  Moreover, \(\mathscr E\), \(\mathscr D_m\), and \(\mathscr D_n\) all have
full column rank. Hence, \(\mathcal C_i\) has full column rank, and therefore $\mathcal C_i^\top\mathcal C_i
\succeq
\sigma_{\min}^2(\mathcal C_i)I_{d_\theta}.$ Using~\eqref{eq:tildephi_factor} and \eqref{mathcalx}, and defining $\mu_i=\bar \mu_i\sigma_{\min}^2(\mathcal C_i) $, it follows that
\begin{equation}\label{tilde_phi}
\begin{aligned}
\sum_{\tau=1}^{h}
\tilde\phi_{i,\tau}^{\top}\tilde\phi_{i,\tau}
&=
\mathcal C_i^\top
\left(
\sum_{\tau=1}^{h}
\mathcal X_{i,\tau}\mathcal X_{i,\tau}^{\top}
\right)
\mathcal C_i\succeq
\mu_i\frac{h^{1-n}\rho_h^{n}}{C_h^{2n-2}}I_{d_\theta}.
\end{aligned}
\end{equation}
Since \(\tilde\phi_{i,h}\) contains only \(n\) block rows of
\(\bar\phi_{i,h}\), it follows that $\sum_{\tau=1}^{h}
\bar\phi_{i,\tau}^{\top}\bar\phi_{i,\tau}
\succeq
\sum_{\tau=1}^{h}
\tilde\phi_{i,\tau}^{\top}\tilde\phi_{i,\tau}\succeq\mu_i\frac{h^{1-n}\rho_h^{n}}{C_h^{2n-2}}I_{d_\theta}$. Taking $\mu=\min_{i\in[L]}\mu_i>0$ and using
$\bar\phi_\tau
=\operatorname{diag}\{\bar\phi_{1,\tau},\ldots,\bar\phi_{L,\tau}\}$
yields the desired bound \eqref{eq: lambda_min_bound}.
\hfill $\square$

{\bfseries\itshape Proof of Lemma \ref{lem4.2}:}
For each $i\in[L]$, define
\[
F_i=
\begin{bmatrix}
A^{T-2}B_i &
A^{T-3}B_i &
\cdots &
B_i &
0_{n\times m}
\end{bmatrix}.
\]
Moreover, define $\mathfrak E
=
\begin{bmatrix}
0_{m\times m(T-2)} & I_m & 0_{m\times m}
\end{bmatrix},$
and
\[
\mathfrak A_i=
\begin{bmatrix}
0 & \cdots &
\underbrace{\mathfrak E}_{i\text{-th block}} &
\cdots & 0\\
F_1 & \cdots & F_i & \cdots & F_L
\end{bmatrix},
\quad
\mathfrak B=
\begin{bmatrix}
0_{m\times n}\\
A^{T-1}
\end{bmatrix}.
\]
Since $A$ is nonsingular, $(A,B_i)$ is controllable, and $T\geq n+2$, we have
\begin{equation}\label{eq:A_i_rank}
\operatorname{rank}(\mathfrak A_i)=m+n.
\end{equation}
Unrolling the system dynamics gives
\begin{equation}\label{eq:zeta_affine}
\zeta_{i,h}(T-1)
=
\begin{bmatrix}
\mathfrak A_i\boldsymbol{U}_h+\mathfrak B\xi_h\\
-1
\end{bmatrix}.
\end{equation}
By~\eqref{first-order} and
$\boldsymbol{P}_h=\hat{\boldsymbol{P}}_h+h^{-\gamma}v_h$,
\[
\boldsymbol{U}_{h}
=
-M^{-1}
\bigl(
\eta_h+\mathcal H\hat{\boldsymbol{P}}_h
\bigr)
-h^{-\gamma}M^{-1}\mathcal Hv_h.
\]
Substituting this expression into~\eqref{eq:zeta_affine} yields
\begin{equation}\label{eq:zeta_decomposition}
\zeta_{i,h}(T-1)
=
\psi_{i,h}+h^{-\gamma}\Delta_{i,h},
\end{equation}
where
\[
\psi_{i,h}
=
\begin{bmatrix}
a_{i,h}\\
-1
\end{bmatrix},
\qquad
\Delta_{i,h}
=
\begin{bmatrix}
r_{i,h}\\
0
\end{bmatrix},
\]
with
\[
\begin{aligned}
a_{i,h}
&=
-\mathfrak A_iM^{-1}
\bigl(
\eta_h+\mathcal H\hat{\boldsymbol{P}}_h
\bigr)
+\mathfrak B\xi_h,\\
r_{i,h}
&=
-D_i v_h,
\qquad
D_i=\mathfrak A_iM^{-1}\mathcal H.
\end{aligned}
\]

Let \(\bar{\mathcal F}_0=\sigma(\xi_1)\), and for each \(h\geq1\), define
\[
\bar{\mathcal F}_h
=
\sigma\bigl(
\xi_1,\ldots,\xi_{h+1},
\epsilon_1,\ldots,\epsilon_h,
v_1,\ldots,v_h
\bigr).
\]
Then $\psi_{i,h}$ is
$\bar{\mathcal F}_{h-1}$-measurable and $\mathbb E[
\Delta_{i,h}\mid\bar{\mathcal F}_{h-1}]
=0.$
Moreover, $\mathcal H$ has full row rank and hence admits a right
inverse. Together with~\eqref{eq:A_i_rank} and the nonsingularity
of $M$, this implies that $D_i$ has full row rank. Thus, for some
$\kappa_i>0$, $D_iD_i^\top\succeq\kappa_i I_{m+n}.$ Since $\mathbb E[v_hv_h^\top]=I$,
\begin{equation}\label{eq:r_covariance}
\mathbb E\!\left[h^{-2\gamma}
r_{i,h}r_{i,h}^\top
\mid\bar{\mathcal F}_{h-1}
\right]
=
h^{-2\gamma}D_iD_i^\top
\succeq
\kappa_i h^{-2\gamma}I_{m+n}.
\end{equation}

Define the conditional Gramian
\begin{equation}\label{eq:def_condition}
G_{i,h}=
\sum_{\tau=1}^h
\mathbb E\!\left[
\zeta_{i,\tau}(T-1)
\zeta_{i,\tau}(T-1)^\top
\mid\bar{\mathcal F}_{\tau-1}
\right].
\end{equation}
For any $x\in\mathbb R^{m+n}$ and $s\in\mathbb R$,
\eqref{eq:zeta_decomposition} and \eqref{eq:r_covariance} give
\begin{equation}\label{eq:conditional_gramian}
\begin{aligned}
\begin{bmatrix}x\\s\end{bmatrix}^{\!\top}
G_{i,h}
\begin{bmatrix}x\\s\end{bmatrix}
\geq{}&
\kappa_i
\sum_{\tau=1}^h\tau^{-2\gamma}\|x\|^2+
\sum_{\tau=1}^h
(a_{i,\tau}^\top x-s)^2.
\end{aligned}
\end{equation}
Let $\bar a_{i,h}=\frac{1}{h}\sum_{\tau=1}^h a_{i,\tau}.$  By Jensen's inequality and \eqref{eq:conditional_gramian},
\[
\begin{aligned}
\begin{bmatrix}x\\s\end{bmatrix}^{\!\top}
G_{i,h}
\begin{bmatrix}x\\s\end{bmatrix}
&\ge
\kappa_i\sum_{\tau=1}^h\tau^{-2\gamma}\|x\|^2
+h(\bar a_{i,h}^\top x-s)^2  \\
&\ge
\underline c_i h^{1-2\gamma}
\Bigl(\|x\|^2+(\bar a_{i,h}^\top x-s)^2\Bigr),
\end{aligned}
\]
for some \(\underline c_i>0\). Moreover, $s^2
\le
2(\bar a_{i,h}^\top x-s)^2
+2\|\bar a_{i,h}\|^2\|x\|^2,$ and $\|\bar a_{i,h}\|
\le
\max_{1\le\tau\le h}\|a_{i,\tau}\|
\le C(1+C_h)$ for some $C>0$. Hence, for some \(\bar C>0\),
\[
\|x\|^2+s^2
\le
\bar C(1+C_h^2)
\Bigl(\|x\|^2+(\bar a_{i,h}^\top x-s)^2\Bigr).
\]
Therefore, for some $\bar c_i>0$,
\[
\begin{bmatrix}x\\s\end{bmatrix}^{\!\top}
G_{i,h}
\begin{bmatrix}x\\s\end{bmatrix}
\ge
\bar c_i
\frac{h^{1-2\gamma}}{1+C_h^2}
\left(\|x\|^2+s^2\right),
\]
which implies
\begin{equation}\label{eq:conditional_gramian_bound}
\lambda_{\min}(G_{i,h})
\ge
\bar c_i
\frac{h^{1-2\gamma}}{1+C_h^2}.
\end{equation}

It remains to compare the sample Gramian with its conditional
expectation. Define
\begin{equation}\label{def_martingale}
\begin{aligned}
    \tilde{G}_{i,h}
=&
\sum_{\tau=1}^h
\left[
\zeta_{i,\tau}(T-1)
\zeta_{i,\tau}(T-1)^\top
\right]-G_{i,h}.
\end{aligned}  
\end{equation}
By \eqref{eq:def_condition}, 
$\{\tilde{G}_{i,h},\bar{\mathcal{F}}_h\}$
forms a martingale sequence. Using~\eqref{eq:zeta_decomposition}, this martingale can be written as
\begin{equation}
\begin{aligned}
\tilde G_{i,h}
={}
&\sum_{\tau=1}^{h}\Bigl[
\tau^{-\gamma}\psi_{i,\tau}\Delta_{i,\tau}^{\top}
+
\tau^{-\gamma}\Delta_{i,\tau}\psi_{i,\tau}^{\top}
\\
&
+
\tau^{-2\gamma}\Delta_{i,\tau}\Delta_{i,\tau}^{\top}
-
\mathbb E\!\left[\tau^{-2\gamma}
\Delta_{i,\tau}\Delta_{i,\tau}^{\top}
\mid\bar{\mathcal F}_{\tau-1}
\right]
\Bigr].
\end{aligned}
\end{equation} 
By condition \eqref{eq:trajectory_bound_aux}, $\zeta_{i,h}(T-1)=O(C_h), a.s.,$ as $h\to \infty$. Since $\{v_h\}$ is bounded, \eqref{eq:zeta_decomposition}
further implies that $\|\psi_{i,h}\|=O(C_h), a.s.$
Consequently,
$\sum_{\tau=1}^{h}\tau^{-2\gamma}\|\psi_{i,\tau}\|^2=O(h^{1-2\gamma}C_h^2), a.s.,$ as $h\to \infty$. Applying Lemma~\ref{lem:martingale} and since $C_h=O(\log h)$, we obtain 
\[
\sum_{\tau=1}^{h}
\tau^{-\gamma}\psi_{i,\tau}\Delta_{i,\tau}^{\top}=O(C_hh^{\frac{1}{2}-\gamma}\sqrt{\log\log h }),\quad a.s.
\]
Moreover, applying Lemma~\ref{lem:martingale} yields
\begin{equation}
\begin{aligned}
&\sum_{\tau=1}^{h}\left(\tau^{-2\gamma}\Delta_{i,\tau}\Delta_{i,\tau}^{\top}
-
\mathbb E\!\left[\tau^{-2\gamma}
\Delta_{i,\tau}\Delta_{i,\tau}^{\top}
\mid\bar{\mathcal F}_{\tau-1}
\right]\right)\\=&O(h^{\frac{1}{2}-\gamma}\sqrt{\log\log h}),\quad a.s.,\quad h\to \infty.
\end{aligned}
\end{equation}
Therefore,
\begin{equation}\nonumber
\|\tilde G_{i,h}\|
=
O\!\left(
C_h
h^{1/2-\gamma}
\sqrt{\log\log h}
\right),
\qquad\text{a.s.}
\end{equation}
Since $C_h=O(\log h)$ and $\gamma<1/2$,
\begin{equation}\label{eq:martingale_analysis}
\|\tilde G_{i,h}\|
=
o\!\left(
\frac{h^{1-2\gamma}}{1+C_h^2}
\right),
\qquad\text{a.s.}
\end{equation}
Finally, by 
\eqref{eq:conditional_gramian_bound}, \eqref{def_martingale}, and \eqref{eq:martingale_analysis},
\[
\lambda_{\min}\!\left(
I_{m+n+1}+\sum_{\tau=1}^h
\zeta_{i,\tau}(T-1)
\zeta_{i,\tau}(T-1)^\top
\right)
\geq
c_i
\frac{h^{1-2\gamma}}{1+C_h^2}
\]
for some constant $c_i>0$, almost surely.
\hfill $\square$
\section{}\label{sec:martingale_proofs}
In this section, we present two martingale convergence results that were used in the convergence analysis of the proposed algorithm above.
\begin{lemma}[Lemma 2 in \cite{wei1985asymptotic}]\label{lem:martingale} 
Let $\left\{\omega_{n}, \mathcal{F}_{n}\right\}$ be a martingale difference sequence satisfying $\sup_{n} 	\mathbb{E}[\|\omega_{n+1}\|^{\alpha}\mid \mathcal{F}_{n}] < \infty$ almost surely for some $\alpha>2$. Let $\left\{f_{n}, \mathcal{F}_{n}\right\}$ be an adapted sequence, and define $s_{n}=\left(\sum_{i=0}^{n}\|f_{i}\|^{2}\right)^{\frac{1}{2}}$. Then, for every $\delta>1/\min\{\alpha,4\}$,
	\begin{equation}\nonumber
		\sum_{i=0}^{n}f_{i}\omega_{i+1} = O(s_{n}(\log s_n)^{\delta}),\quad a.s.
	\end{equation}
	Furthermore, if $\|f_n\|=o(s_n^{\epsilon})$ almost surely for some $0<\epsilon<1$,
    then
	\begin{equation}\nonumber
		\sum_{i=0}^{n}f_{i}\omega_{i+1} = O(s_{n}\sqrt{\log\log s_n}),\quad a.s.
	\end{equation}
\end{lemma}
\begin{lemma}\label{lem:excitation_equivalence}
Let 
\(\{\omega_h,\mathcal F_h\}\) be a martingale difference
sequence and \(\{f_h,\mathcal F_h\}\) an adapted sequence, satisfying
\[
\mathbb E[\omega_{h+1}^\top\omega_{h+1}\mid\mathcal F_{h}]
=\Sigma,\;\; \sup_{h\geq 0}
\mathbb E\left[
\|\omega_{h+1}\|^{\alpha}
\mid
\mathcal F_{h}
\right]
<\infty,\;\;\text{a.s.},
\]
for some \(\alpha>4\). Define
$\bar f_h=f_h+\omega_{h+1}$ for all $h\geq 0$,
and
\[
S_h
=\sum_{\tau=0}^h
f_\tau^\top f_\tau,
\quad
\bar S_h
=\sum_{\tau=0}^h
\bigl(\bar f_\tau^\top\bar f_\tau-\Sigma\bigr).
\]
Then, for every fixed $\eta\in (1/2, 1)$, 
\begin{equation}\nonumber
\|S_h-\bar S_h\|
=o\left((\lambda_{\max}(S_h)+h)^{\eta}\right),
\quad\text{a.s.}
\end{equation}
\end{lemma}
\begin{proof}
By the definitions of $\bar f_h$, $S_h$, and $\bar S_h$,
\[
\bar S_h-S_h
=\sum_{\tau=0}^h
\left(
f_\tau^\top\omega_{\tau+1}
+
\omega_{\tau+1}^\top f_\tau
\right)
+
\sum_{\tau=0}^h
\left(
\omega_{\tau+1}^\top\omega_{\tau+1}-\Sigma
\right).
\]
Since $(\sum_{\tau=0}^{h}\|f_{\tau}\|^{2})^{1/2}=O(\sqrt{\lambda_{\max}(S_h)})$ almost surely, by Lemma~\ref{lem:martingale},
\[
\sum_{\tau=0}^h
f_\tau^\top\omega_{\tau+1}
=O\left(
\sqrt{
\lambda_{\max}(S_h)}
(\log\lambda_{\max}(S_h)
)^{1/2}
\right),
\quad\text{a.s.}
\]
Moreover,
\(\{\omega_{h}^\top\omega_{h}-\Sigma, \mathcal F_h\}\) is a martingale
difference sequence with uniformly bounded conditional
\(\alpha/2\)-th moments. Since \(\alpha/2>2\), by Lemma~\ref{lem:martingale},
\[
\sum_{\tau=0}^h
(\omega_{\tau+1}^\top\omega_{\tau+1}-\Sigma)
=O\left(\sqrt{h\log\log h}\right),
\qquad\text{a.s.}
\]
This completes the proof of Lemma~\ref{lem:excitation_equivalence}. 
\end{proof}

\end{document}